\documentclass[journal]{IEEEtran}
\usepackage{cite}
\usepackage{amsmath,amssymb,amsfonts,amsthm}

\usepackage[utf8]{inputenc}
\usepackage[T1]{fontenc}

\usepackage{graphicx}
\usepackage{textcomp}
\usepackage{xcolor,soul}
\usepackage{algpseudocode}
\usepackage{amsmath}
\usepackage{mathrsfs}
\usepackage{bm}
\usepackage[abs]{overpic}
\usepackage[nocomma]{optidef}

\usepackage{physics}
\usepackage{tikz}
\usepackage{mathdots}
\usepackage{yhmath}
\usepackage{cancel}
\usepackage{siunitx}
\usepackage{array}
\usepackage{multirow}
\usepackage{amssymb}
\usepackage{gensymb}
\usepackage{tabularx}
\usepackage{extarrows}
\usepackage{booktabs}
\usepackage[ruled,linesnumbered]{algorithm2e}
\makeatletter
\newcommand{\mathleft}{\@fleqntrue\@mathmargin0pt}
\newcommand{\mathcenter}{\@fleqnfalse}
\makeatother

\usetikzlibrary{fadings}
\usetikzlibrary{patterns}
\usetikzlibrary{shadows.blur}
\usetikzlibrary{shapes}

\theoremstyle{plain}
\newtheorem*{theorem*}{Theorem}

\newtheorem{proposition}{Proposition}
\newtheorem{remark}{Remark}
\newtheorem*{eg}{Example}

\usepackage{xcolor,varwidth}

\usepackage[scaled=.75]{beramono}
\newcommand{\bpara}[1]		{\medskip \noindent {\bf #1}}

\renewcommand\geq\geqslant
\renewcommand\leq\leqslant

\newcommand\fig[1]			{Fig.~\ref{#1}}

\def\eg					    {\emph{e.g.,}~}
\def\ie					    {\emph{i.e.,}~}

\def\R{\mathbb{R}}

\def\C{\mathbb{C}}

\def\bW{\mathbf{W}}

\def\Re{\mathcal{R}}

\def\bF{\mathbf{F}}

\def\bthe{\boldsymbol{\Theta}}

\def\bY{\mathbf{Y}}

\usepackage{xspace}

\DeclareDocumentCommand{\sd}{o}
{{\underline{\ast}\IfValueT{#1}{_{#1}}}}

\DeclareDocumentCommand{\xmod}{m o o o}
{%
\IfNoValueTF{#4}
{{#1}\IfValueT{#2}{_{m,\mathsf{#2}}}\IfValueT{#3}{#3}}
{{#1}\IfValueT{#2}{_{m,\mathsf{#2}}^{\mathsf{#4}}}\IfValueT{#3}{#3}}
}

\renewcommand{\Re}{\mathfrak{R}}

\usepackage[inline]{enumitem}
\DeclareMathAlphabet{\mathsfit}{T1}{\sfdefault}{\mddefault}{\sldefault}
\SetMathAlphabet{\mathsfit}{bold}{T1}{\sfdefault}{\bfdefault}{\sldefault}
\usepackage[switch]{lineno}
\usepackage{siunitx}

\def\BibTeX{{\mathrm B\kern-.05em{\sc i\kern-.025em b}\kern-.08em
    T\kern-.1667em\lower.7ex\hbox{E}\kern-.125emX}}

\usepackage{hyperref}
\usepackage{etoolbox}

\begin{document}
\bstctlcite{IEEEexample:BSTcontrol}

\title{Microwave Linear Analog Computer (MiLAC)-Aided MIMO Radar Sensing: Transmit Beamforming Design and DoA Estimation}

\author{
Ziang Liu, Zheyu Wu,
        and~Bruno Clerckx,~\IEEEmembership{Fellow,~IEEE}
\thanks{Z. Liu, Z. Wu, and B. Clerckx are with the Communications \& Signal Processing (CSP) Group at the Dept. of Electrical and Electronic Engg., Imperial College London, SW7 2AZ, UK. (e-mail: ziang.liu20@imperial.ac.uk).
}}

\maketitle

\begin{abstract}
Multiple-input multiple-output (MIMO) radar benefits from waveform diversity and large spatial degrees of freedom (DoFs), making it attractive for high-resolution sensing. Scaling MIMO radar to massive arrays can further improve sensing performance, but it also increases hardware cost, power consumption, and digital processing complexity. The microwave linear analog computer (MiLAC) can tackle these challenges by moving linear operations from the digital domain to the analog domain. {MiLAC has shown promising benefits for communications in recent studies and this paper identifies its potential for radar sensing.} Specifically, we consider both MiLAC-aided transmit beamforming and receiver-side two-dimensional discrete Fourier transform (2D-DFT)-based direction-of-arrival (DoA) estimation. For transmit beamforming, we formulate a weighted Cram\'{e}r--Rao bound (CRB) minimization problem under lossless and reciprocal MiLAC constraints and solve a reduced covariance semidefinite program (SDP) followed by closed-form MiLAC realization. We further prove that MiLAC-aided and fully-digital beamforming realize the same set of transmit covariance matrices for any given number of waveforms and transmit power. Therefore, both achieve the same optimal CRB. For targets with unknown elevation and azimuth angles, we also show that achieving the optimal CRB needs at most three times as many waveforms as targets. Since a MiLAC requires one radio-frequency (RF) chain per waveform, it can achieve the optimal CRB with fewer RF chains than fully-digital beamforming.
For receiver processing, we show that the 2D-DFT can be implemented by a lossless reciprocal MiLAC, which enables analog-domain DoA estimation without digital optimization.
Numerical results confirm the theoretical findings and show that the MiLAC-aided approach achieves the same CRB and DoA estimation performance as the fully-digital benchmark. {Meanwhile, hardware cost and power consumption are reduced because only low-resolution digital-to-analog converters (DACs) are required at the transmitter, while fewer RF chains and analog-to-digital converters (ADCs) are needed at the receiver. Moreover, performing the 2D-DFT in the analog domain eliminates all digital DFT operations for DoA estimation.
}
\end{abstract}

\begin{IEEEkeywords}
Analog computation, Cram\'{e}r--Rao bound (CRB), direction-of-arrival (DoA) estimation, Microwave linear analog computer (MiLAC), MIMO radar.
\end{IEEEkeywords}

\section{Introduction}
Sensing is expected to be a core functionality of future wireless systems \cite{gonzalez2024integrated}. Among sensing technologies, multiple-input multiple-output (MIMO) radar is particularly attractive because multiple antennas and multiple orthogonal waveforms provide large spatial degrees of freedom (DoFs) and waveform diversity compared with phased-array radar \cite{li2007mimo, stoica2007probing}. These DoFs bring gains in beampattern design, target detection, parameter estimation \cite{stoica2007probing,fuhrmann2008transmit}, and interference suppression \cite{tang2015relative}. As a result, MIMO radar has been adopted in applications such as automotive radar \cite{sun2020mimo}, remote sensing \cite{tarchi2012mimo}, and integrated sensing and communication \cite{liu2023joint}. Extending MIMO radar to massive arrays is a natural step toward higher spatial resolution and larger array gain \cite{bjornson2019massive, fortunati2020massive}, and it has motivated substantial recent interest in mm-wave \cite{li2021fast} and terahertz (THz) \cite{elbir2021terahertz} massive MIMO radar systems.

However, this scaling creates two bottlenecks. The first is \textit{hardware cost} and \textit{power consumption}. Fully exploiting the DoFs of a MIMO radar requires one radio-frequency (RF) chain per antenna, including digital-to-analog converters (DACs) and analog-to-digital converters (ADCs), which dominate the overall power consumption \cite{zhang2016spectral}. For example, a 256-antenna array using 12-bit ADCs at 20 Gsamples/s can consume 256 W \cite{zhang2018low}, which is difficult to sustain in practical deployments. The other bottleneck is \textit{computational complexity}. As the array size increases, the computational complexity of digital transmit beamforming and receiver-side parameter estimation grows substantially. Even discrete Fourier transform (DFT)-based digital processing leads to non-negligible complexity. In addition, subspace-based methods such as multiple signal classification (MUSIC)
additionally requires eigenvalue decompositions and matrix inversions, resulting in complexity at the order of $\mathcal{O}(N^3)$ \cite{schmidt1986multiple}. For latency-sensitive applications such as automotive radar, this digital processing burden can become prohibitive.

Existing approaches typically address only part of this problem. To reduce hardware cost and power consumption, low-resolution DACs/ADCs architectures have been studied, because the power consumption of DACs/ADCs scales exponentially with quantization resolution \cite{zhang2016spectral}. At the extreme, 1-bit DACs/ADCs substantially simplify the transceiver, and have been considered for MIMO radar waveform design, detection, parameter estimation, and tracking \cite{cheng2019transmit, xi2020gridless, stein2015asymptotic}. While such designs can be attractive, the information loss introduced by coarse quantization is irreversible and can degrade sensing performance. Hybrid analog-digital beamforming is another route by reducing the number of RF chains, but the resulting reduction in waveform diversity and beamforming flexibility causes performance loss compared to fully-digital architectures \cite{shu2018low}. Another alternative is to move beamforming entirely to the analog domain, thus allowing the RF chains to directly carry radar-modulated signals, \eg frequency-modulated continuous-wave (FMCW) waveforms. Therefore, only low-resolution DACs/ADCs are required, which can reduce hardware cost and power consumption.
Stacked intelligent metasurfaces (SIMs), which are realized by cascaded transmissive RIS layers, are a representative example of this direction \cite{an2024stacked}. Specifically, the 2D-DFT-based DoA estimation is achieved by SIMs and directly in the analog domain \cite{an2024stacked}.

Analog-domain processing is also advantageous from a computational complexity perspective. Over-the-air computation (OAC) exploits the property of electromagnetic (EM) waveform superposition to realize efficient computation during transmission. In addition, this method has been employed for data aggregation in integrated sensing and communication systems \cite{li2023over}. However, its supported operations are mainly limited to nomographic functions such as weighted sums and arithmetic means \cite{li2023over}. For more general linear transforms, SIM-aided analog computing has recently been proposed for tasks such as transmit beampattern design and DFT realization \cite{niu2024stacked, an2024two}.
However, analog domain processing typically decreases the flexibility and thus leads to performance loss compared to digital processing due to the operations are on the microwave impedance network \cite{nerini2025analog1}.

\begin{figure}[t]
    \centering
    \includegraphics[width = 0.48\textwidth]{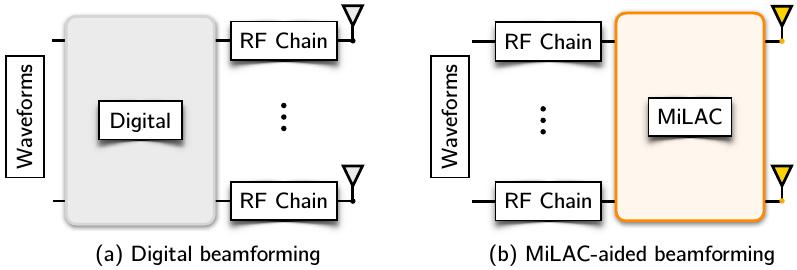}
    \centering
    \caption{Comparison between (a) digital beamforming and (b) MiLAC-aided analog beamforming.}
    \label{fig:digital_vs_milac}
\end{figure}

{So far a unified framework of analog computation is missing, but microwave linear analog computer (MiLAC) provides a first promising attempt \cite{nerini2025analog1}.} Built on multiport microwave networks with reconfigurable impedance components, MiLAC implements linear transformations directly in the analog domain (cf.\ \fig{fig:digital_vs_milac}) by tuning the impedance network \cite{nerini2025analog1, nerini2025analog2}. Recent work has shown that MiLAC supports purely analog linear zero-forcing beamforming (ZFBF) and linear minimum mean square error (LMMSE)-based transmit and receive beamforming for MIMO communications \cite{nerini2025analog1, nerini2025analog2}. {A physics-compliant model considering mutual coupling is investigated in \cite{nerini2026physics}, showing that mutual coupling is on average beneficial in MiLAC-aided systems. To alleviate the high circuit complexity and control overhead, \cite{nerini2026mimo} proposes a one-layer stem-connected MiLAC architecture that preserves capacity-achieving performance in single-user MIMO systems while reducing hardware complexity. Analog-domain channel estimation is proposed in \cite{zhang2026channel}, achieving the same performance as its digital-domain counterpart while reducing computational complexity and hardware requirements.}

From a hardware perspective, losslessness and reciprocity are common constraints for microwave impedance networks. Under these constraints, MiLAC-aided beamforming can match the capacity of fully-digital beamforming in point-to-point MIMO communication systems \cite{nerini2025capacity}. For multi-user communications, \cite{wu2026microwave, fang2026performance} further show that MiLAC does not in general retain the full flexibility of digital beamforming, but a hybrid digital-MiLAC architecture can achieve maximum performance with only $K$ RF chains for $K$ users. In contrast, conventional hybrid beamforming requires $2K$ \cite{sohrabi2016hybrid}.
More recently, \cite{zhou2026two} shows that a two-layer reciprocal and lossless MiLAC achieves the same performance as fully-digital beamforming in multi-user MISO systems, and \cite{peng2026hybrid} studies hybrid digital-MiLAC beamforming for wideband multi-user (MU) MIMO systems, demonstrating that the proposed design consistently outperforms conventional hybrid digital-analog beamforming.
To summarize, MiLAC-aided communications has shown the following promising benefits:
\textit{1) MiLAC can achieve the same performance as digital beamforming.}
\textit{2) MiLAC requires only as many RF chains as data streams.}
Since the signals carried and sampled at the RF chains are the actual symbols, \textit{3) low-resolution ADCs and DACs are sufficient.}
\textit{4) MiLAC does not require any computation at each symbol time}, since the symbols are precoded and combined fully in the analog domain.
\textit{5) MiLAC can perform ZFBF and MMSE combining with reduced computational complexity, \ie $\mathcal{O}(N^2)$ instead of $\mathcal{O}(N^3)$.}

However, these results focus on communication systems. The use of MiLAC for MIMO radar sensing, especially under practical lossless and reciprocal constraints, remains unexplored. Recent prototype-level evidence shows that hybrid couplers and phase shifters can realize linear transforms, \ie a $4 \times 4$ DFT \cite{nerini2026analog}, further motivating the study of MiLAC for radar sensing. {In radar systems, scaling to massive arrays increases not only hardware cost and power consumption due to high-resolution DACs and RF chains, but also computational complexity due to large-dimensional DFT operations for DoA estimation.} To address these challenges, this paper investigates MiLAC-aided MIMO radar for the first time from both the transmitter and receiver sides. Under lossless and reciprocal constraints, we show that MiLAC-aided radar sensing has the following benefits.
\textit{1) MiLAC achieves the same sensing performance as fully-digital beamforming}, since the set of achievable transmit covariance matrices under MiLAC constraints is identical to that of fully-digital beamforming, and the Cram\'{e}r--Rao bound (CRB) depends on the beamforming matrix only through the transmit covariance (See Section \ref{subsec:crb_equiv}). In addition, the 2D-DFT-based DoA estimation can be exactly implemented by a lossless reciprocal MiLAC (See Section \ref{sec:receiver}).
\textit{2) MiLAC reduces the hardware cost and power consumption of RF chains}, because MiLAC enables fully analog beamforming and thus the RF chains can directly carry radar-modulated signals with lower-resolution DACs.
\textit{3) MiLAC requires no per-symbol digital processing}, as the transmit beamforming in MIMO radar is entirely implemented in the analog domain.
\textit{4) DFT-based DoA estimation can be achieved with reduced computational complexity}, because the DFT can be exactly implemented by a lossless reciprocal MiLAC without any offline optimization.
\textit{5) Low DFT latency especially in massive MIMO radar}, because the DFT is performed in the analog domain. This is particularly beneficial for real-time detection and tracking applications such as automotive radar.
These advantages make MiLAC especially appealing for massive array and real-time radar applications.

{Benefits 2) and 3) were first shown for MiLAC-aided communications \cite{nerini2025analog1, nerini2025analog2} and also hold for MiLAC-aided radar sensing. Benefits 1), 4), and 5) are new in this paper.\par}

\bpara{Contributions and Overview of Results.} The main contributions of this paper are summarized as follows:
\begin{enumerate}
    \item This is the first work to investigate MiLAC-aided MIMO radar sensing under practical lossless and reciprocal constraints. Specifically, we formulate transmit beamforming for MiLAC-aided MIMO radar as a weighted CRB minimization problem for multi-target 2D angle estimation, under both the total transmit power constraint and the lossless reciprocal constraints imposed by the MiLAC impedance network. On the receiver side, we show that the 2D-DFT can be implemented by a lossless reciprocal MiLAC, which enables analog-domain DoA estimation without offline digital optimization.

    \item We prove that MiLAC-aided and fully-digital beamforming with the same number of waveforms achieve the same optimal CRB. The key insight is that the CRB depends on the beamforming matrix  only through the transmit covariance, and we show that the set of achievable transmit covariance matrices is identical under both architectures. In addition, the number of waveforms can be smaller than the number of transmit antennas. Since a MiLAC needs one RF chain per waveform, it can use fewer RF chains than transmit antennas.

    \item Building on this covariance equivalence, we reformulate the weighted CRB minimization problem as a transmit covariance optimization problem. Considering a system with $N$ transmit antennas and $K$ targets, we prove that this problem with $\min(N,3K)$ waveforms has the same optimal weighted CRB as with $N$ waveforms and can be globally solved by a reduced semidefinite program (SDP). Therefore, a MiLAC can use fewer RF chains than transmit antennas without any loss in sensing performance, which reduces hardware cost and power consumption. Finally, we construct in closed form the power-allocation vector and the symmetric and unitary MiLAC scattering matrix that realize the optimized covariance.

    \item Numerical results confirm that the proposed MiLAC-aided beamforming design achieves the same CRB as fully-digital beamforming and can probe the beams at the target angles. On the receiver side, MiLAC-aided 2D-DFT DoA estimation achieves the same performance as fully-digital DFT-based DoA estimation. Furthermore, the receiver-side computational complexity of 2D-DFT-based DoA estimation is reduced by MiLAC, especially for large arrays. We also show that a MiLAC does not need one RF chain per transmit antenna to achieve the optimal weighted CRB. Fewer RF chains than transmit antennas suffice.

    \item {Overall, under lossless and reciprocal constraints, we show that (a) MiLAC achieves the same sensing performance as fully-digital beamforming, (b) MiLAC reduces the hardware cost and power consumption of RF chains, (c) MiLAC requires no per-symbol digital processing, (d) DFT-based DoA estimation can be achieved with reduced computational complexity, and (e) low DFT latency especially in massive MIMO radar.}

\end{enumerate}

\bpara{Organization of This Paper.} Section \ref{sec:sys} presents the system model. Section \ref{sec:problem_formulation} explains the transmit-side beamforming design. Section \ref{sec:receiver} addresses receiver-side MiLAC-aided 2D-DFT-based DoA estimation. Numerical results are given in Section \ref{sec:simu}, and conclusions are drawn in Section \ref{sec:con}.

\bpara{Notation.} Matrices, vectors, and scalars are represented by bold uppercase, bold lowercase, and regular fonts, respectively. The sets of real, nonnegative real, and complex numbers are denoted by $\mathbb{R}$, $\mathbb{R}_+$, and $\mathbb{C}$, respectively. The transpose, conjugate-transpose, complex conjugate, and inverse are denoted by $(\cdot)^\top$, $(\cdot)^H$, $(\cdot)^*$, and $(\cdot)^{-1}$, respectively. The $(i,j)$th element of $\mathbf{X}$ and its submatrix with rows $i_1$ to $i_2$ and columns $j_1$ to $j_2$ are denoted by $[\mathbf{X}]_{i,j}$ and $[\mathbf{X}]_{i_1:i_2,j_1:j_2}$, respectively. The identity matrix, all-zero matrix, and all-ones vector are denoted by $\mathbf{I}$, $\mathbf{0}$, and $\mathbf{1}$, respectively. The real part, imaginary part, absolute value, and imaginary unit are denoted by $\Re(\cdot)$, $\mathfrak{I}(\cdot)$, $|\cdot|$, and $\jmath$, respectively. The diagonal matrix, trace, rank, span, dimension, statistical expectation, and Kronecker product are denoted by $\operatorname{diag}(\cdot)$, $\trace(\cdot)$, $\operatorname{rank}(\cdot)$, $\operatorname{span}\{\cdot\}$, $\dim(\cdot)$, $\mathbb{E}(\cdot)$, and $\otimes$, respectively. The spectral norm (Euclidean norm for vectors), Frobenius norm, ceiling function, and circularly symmetric complex Gaussian distribution are denoted by $\|\cdot\|_2$, $\|\cdot\|_F$, $\lceil\cdot\rceil$, and $\mathcal{CN}(\cdot,\cdot)$, respectively. Vector inequalities hold element-wise, and $\mathbf{X}\succeq\mathbf{Y}$ means that $\mathbf{X}-\mathbf{Y}$ is positive semidefinite. For $y>0$, $\operatorname{mod}(x,y)\in[0,y)$ denotes the remainder of $x/y$.

\begin{figure}[t]
    \centering
    \includegraphics[width=\linewidth]{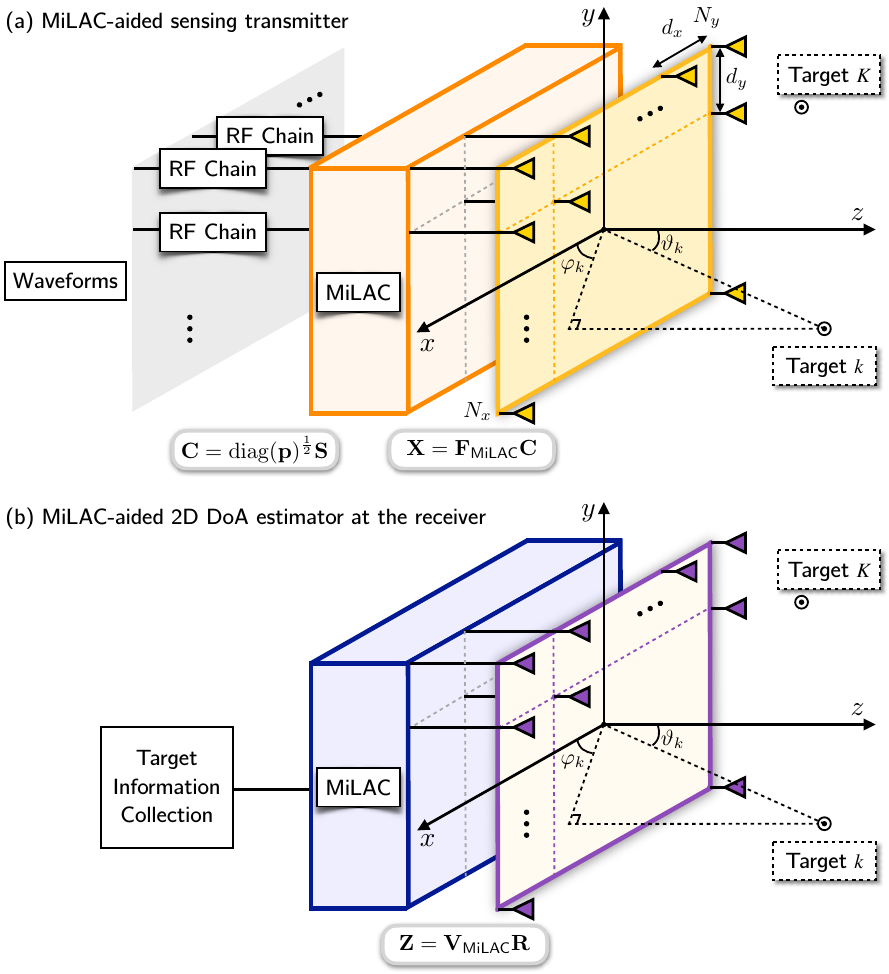}
    \centering
    \caption{Architecture of the MiLAC-aided sensing system.}
    \label{fig:array}
\end{figure}
\begin{figure}[!t]
    \centering
    \includegraphics[width = 0.3\textwidth]{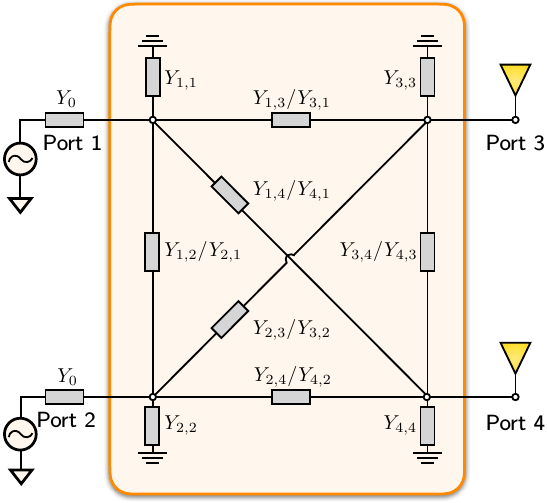}
    \centering
    \caption{A 4-port MiLAC with 2 input ports and 2 output ports. {Reciprocal MiLAC means $Y_{i,j} = Y_{j,i}$.}}
    \label{fig:milac_port}
\end{figure}
\section{System Model}
\label{sec:sys}
As shown in \fig{fig:array}, the co-located transmit and receive antenna uniform planar arrays (UPA) are assumed to have $N = N_x N_y$ antennas, where $N_x$ and $N_y$ are the number of elements in $x$-direction and $y$-direction, respectively. Let $M$ denote the number of orthogonal waveforms, with $1\leq M\leq N$, and let $L\geq M$ denote the number of snapshots.

\subsection{MiLAC Model}
\label{subsec:milac_model}

We consider a fully-connected MiLAC (cf. \fig{fig:milac_port}), where each pair of ports is connected through a reconfigurable impedance and each port is also connected to ground through a reconfigurable impedance \cite{nerini2025analog1,nerini2025analog2,wu2026microwave}. The input-output relationship of the $(N+M)$-port MiLAC is characterized by $\mathbf{F}_\mathsf{MiLAC} \in \mathbb{C}^{N \times M}$, and can be expressed as the scattering matrix $\bthe \in \mathbb{C}^{(N+M) \times (N+M)}$. The scattering matrix relates the incident and reflected waves at all ports and is obtained from the admittance matrix $\bY \in \mathbb{C}^{(N+M) \times (N+M)}$ via \cite{nerini2025capacity}
\begin{equation}
    \label{Eq:theta_Y}
    \bthe=(\mathbf{I}+Z_0\bY)^{-1}(\mathbf{I}-Z_0\bY),
\end{equation}
where $Z_0 = 50\;\Omega$ is the characteristic impedance. Following \cite{wu2026microwave}, the transfer matrix $\mathbf{F}_\mathsf{MiLAC}$ is extracted from $\bthe$ as
\begin{equation}
\label{eq:transfer_matrix}
\bF_{\mathsf{MiLAC}}=\frac{1}{2}[\bthe]_{M+1:M+N,1:M}.
\end{equation}

Subsequently, we impose two physical-compliant constraints on the microwave network, \ie \textit{Losslessness} and \textit{Reciprocity}. Specifically, losslessness requires that the network dissipates no energy, so that $\bthe$ is unitary. Reciprocity requires that the impedance between any two ports is identical in both directions, so that $\bthe$ is symmetric. Therefore, we have the following constraints:
\begin{equation}
\bthe^H\bthe=\mathbf{I}, \quad \bthe=\bthe^\top.
\end{equation}

\subsection{Transmit-side MiLAC-aided Beamforming}
\label{subsec:milac_bf}
$\mathbf{F}_\mathsf{MiLAC} \in \mathbb{C}^{N \times M}$ is the input-output transfer matrix of the MiLAC. The transmitted signal is $\mathbf{X} = \mathbf{F}_\mathsf{MiLAC} \mathbf{C}$, where $\mathbf{C} = \operatorname{diag}(\mathbf{p})^\frac{1}{2} \mathbf{S} \in \mathbb{C}^{M \times L}$ is the source signal at the $M$ RF chains. Here, $\mathbf{p} \in \mathbb{R}^{M \times 1}$ with $p_m \geq 0$ for $m=1,\ldots,M$ denotes the power-allocation vector for the $M$ waveforms, and {$\mathbf{S} \in \mathbb{C}^{M \times L}$ is the known deterministic waveform matrix satisfying $\mathbf{S} \mathbf{S}^H = L\mathbf{I}$ \cite{bekkerman2006target}. Each row of $\mathbf{S}$ is one waveform. Specifically, $\mathbf{S}$ can be formed using phase-coded waveforms, where each waveform has constant modulus and different waveforms are orthogonal to one another, \eg using Hadamard codes \cite{li2007mimo}.}

Using \eqref{eq:transfer_matrix} and the lossless-reciprocal constraints, the set of achievable MiLAC-aided beamforming matrices is \cite{wu2026microwave}
\begin{equation}
\label{Wmilac0}
\begin{aligned}
\mathcal{W}_{\mathsf{MiLAC}}&=\big\{\bW=\mathbf{F}\text{diag}(\mathbf{p})^{\frac{1}{2}}\mid\bF=[\bthe]_{M+1:M+N,1:M},\\
&\bthe=\bthe^\top,~\bthe^H\bthe=\mathbf{I},~\mathbf{1}^\top\mathbf{p}\leq P_T\big\},
\end{aligned}
\end{equation}
where $\bF \triangleq 2\bF_{\mathsf{MiLAC}}$ is introduced to simplify notation, and the factor of two is absorbed into the channel normalization.
As shown in \cite{wu2026microwave}, the lossless and reciprocal constraints on the full scattering matrix $\bthe$ reduce to a spectral-norm constraint on the submatrix $\bF$, yielding the equivalent characterization
\begin{equation}
\label{MiLAC2}
\begin{aligned}
\mathcal{W}_{\mathsf{MiLAC}}=\big\{\bW&=\mathbf{F}\text{diag}(\mathbf{p})^{\frac{1}{2}}\mid \|\bF\|_2\leq 1,~\mathbf{1}^\top\mathbf{p}\leq P_T\big\}.
\end{aligned}
\end{equation}
If we focus on the beamforming matrix $\bW$, based on the Lemma 1 in \cite{wu2026microwave}, the set can be equivalently transformed to
\begin{equation}\label{MiLAC3}
\begin{split}
\mathcal{W}_{\mathsf{MiLAC}}=\{\bW\mid{}&\exists\mathbf{p}\in\R_+^{M\times1}:\ \bW^H\bW\preceq\text{diag}(\mathbf{p}),\\
&\mathbf{1}^\top\mathbf{p}\leq P_T\}.
\end{split}
\end{equation}
Utilizing such set of MiLAC-aided beamforming matrices, we can design the transmit beamforming for MIMO radar sensing as shown in \fig{fig:array}(a).

\subsection{Sensing System Model}
The element spacings are $d_x$ and $d_y$. Let $K$ denote the total number of targets. For the $k^{\mathrm{th}}$ target (where $k = 1, \dots, K$), $\varphi_k \in[0,2 \pi)$ and $\vartheta_k \in[0, \pi / 2]$ denote the azimuth angle and elevation angle of the signal relative to the antenna arrays. Thus, the electrical angles $\psi_{x,k}$ and $\psi_{y,k}$ in the $x$- and $y$-directions for the $k^{\mathrm{th}}$ target are given by
\begin{equation}
\label{eq:psi_xy}
    \psi_{x,k}=\kappa d_x \sin (\vartheta_k) \cos (\varphi_k), \quad
    \psi_{y,k}=\kappa d_y \sin (\vartheta_k) \sin (\varphi_k),
\end{equation}
where $\kappa = 2 \pi/\lambda$ denotes the wavenumber and $\lambda$ is the wavelength.

The steering vectors with respect to (w.r.t.) the transmit and receive antenna arrays for the $k^{\mathrm{th}}$ target are represented by $\mathbf{a}(\psi_{x,k}, \psi_{y,k}) \in \mathbb{C}^{N \times 1}$ and $\mathbf{b}(\psi_{x,k}, \psi_{y,k}) \in \mathbb{C}^{N \times 1}$, respectively \cite{heidenreich2012joint}, where
\begin{equation}
\mathbf{a}(\psi_{x,k}, \psi_{y,k}) \triangleq \mathbf{a}_y(\psi_{y,k}) \otimes \mathbf{a}_x(\psi_{x,k})
\end{equation}
\begin{equation}
\mathbf{b}(\psi_{x,k}, \psi_{y,k}) \triangleq  \mathbf{a}_y(\psi_{y,k}) \otimes \mathbf{a}_x(\psi_{x,k}).
\end{equation}
The elements of the vectors $\mathbf{a}_x(\psi_{x,k}) \in \mathbb{C}^{N_x \times 1}$ and $\mathbf{a}_y(\psi_{y,k}) \in \mathbb{C}^{N_y \times 1}$ are defined as follows:
\begin{equation}
[\mathbf{a}_x(\psi_{x,k})]_{n_x} \triangleq e^{\jmath \psi_{x,k}(n_x-1)}, \quad n_x=1, \dots, N_x,
\end{equation}
\begin{equation}
[\mathbf{a}_y(\psi_{y,k})]_{n_y} \triangleq e^{\jmath \psi_{y,k}(n_y-1)}, \quad n_y=1, \dots, N_y.
\end{equation}
The $n^{\rm{th}}$ antenna element has the coordinates $\left(n_x, n_y\right)$, where $n_x$ and $n_y$ are the indices of the antenna element in the $x$- and $y$-directions, respectively, given by
\begin{equation}
\label{eq:nxny}
    n_{y} \triangleq\left\lceil n / N_{x}\right\rceil, \quad
    n_{x} \triangleq n-\left(n_{y}-1\right) N_{x}.
\end{equation}

Let $\mathbf{A}_k(\psi_{x,k}, \psi_{y,k}) \triangleq \mathbf{b}(\psi_{x,k}, \psi_{y,k}) \mathbf{a}^\top(\psi_{x,k}, \psi_{y,k})$ denote the round-trip steering matrix of target $k$. The received signal $\mathbf{R} \in \mathbb{C}^{N \times L}$ at the receive UPA, after reflection from the $K$ targets, can be expressed as
\begin{equation}
\label{eq:rx_signal}
\begin{aligned} \mathbf{R}
& \! = \! \! \sum_{k=1}^{K} \! \alpha_{r,k} \mathbf{b}(\psi_{x,k}, \! \psi_{y,k}) \mathbf{a}\!^\top \! (\psi_{x,k}, \! \psi_{y,k}) \mathbf{F}_\mathsf{MiLAC} \operatorname{diag}(\mathbf{p})\!^\frac{1}{2} \mathbf{S} \! + \! \mathbf{N}\\
& = \! \! \sum_{k=1}^{K} \alpha_{r,k} \mathbf{A}_k(\psi_{x,k}, \psi_{y,k})\mathbf{F}_\mathsf{MiLAC} \mathbf{C} + \mathbf{N},
\end{aligned}
\end{equation}
where $\alpha_{r,k}$ denotes the reflection coefficient of the $k^{\mathrm{th}}$ target, which captures both the round-trip path loss and the radar cross-section (RCS), as well as the carrier phase shift $e^{-\jmath2\pi f_c\tau_k}$ caused by the round-trip delay $\tau_k$ at carrier frequency $f_c$. We use a narrowband model in which the round-trip delays of all targets differ by much less than one snapshot interval. The target parameters remain constant over the observation. The additive white Gaussian noise (AWGN) matrix is denoted by $\mathbf{N} \triangleq \left[\mathbf{n}[1], \mathbf{n}[2], \ldots, \mathbf{n}[L]\right] \in \mathbb{C}^{N \times L}$, where $\mathbf{n}[t]\sim\mathcal{CN}(\mathbf{0},\sigma_r^2\mathbf{I})$ and $\sigma_r^2$ denotes the noise variance per receive antenna.

\subsection{Receiver-side MiLAC-aided Signal Processing}
Beyond transmit beamforming, MiLAC can also perform receiver-side processing as shown in \fig{fig:array}(b). A second MiLAC with transfer matrix $\mathbf{V}_\mathsf{MiLAC} \in \mathbb{C}^{N \times N}$ is placed at the receiver and designed to implement the 2D-DFT for DoA estimation. The output powers are detected and averaged over blocks of snapshots.
The resulting output is
\begin{equation}
\begin{aligned}
\label{eq:rx_milac_output}
    \mathbf{Z} & = \mathbf{V}_\mathsf{MiLAC} \mathbf{R} \\
    & = \sum_{k=1}^{K}\alpha_{r,k}\mathbf{V}_\mathsf{MiLAC}\mathbf{A}_k(\psi_{x,k},\psi_{y,k})\mathbf{F}_\mathsf{MiLAC}\mathbf{C}+\mathbf{V}_\mathsf{MiLAC}\mathbf{N}.
\end{aligned}
\end{equation}
Considering the lossless and reciprocal constraints, the feasible set for the receiver-side MiLAC is \cite{ nerini2025analog1, nerini2025analog2,nerini2025capacity, wu2026microwave}
\begin{equation}
\label{eq:milac_2DDFT}
\left\{\mathbf{V}_\mathsf{MiLAC} \rvert \mathbf{V}_\mathsf{MiLAC} \! = \! \frac{1}{2}[\boldsymbol{\Theta}]_{N+1: 2N, 1: N}, \boldsymbol{\Theta}\!=\!\boldsymbol{\Theta}^\top\!\!, \boldsymbol{\Theta}^H \boldsymbol{\Theta}\!= \!\mathbf{I} \right\}
\end{equation}

\section{Transmit-side MiLAC-aided Beamforming: Problem Formulation and Solution}
\label{sec:problem_formulation}
In this section, we first derive the CRB for 2D multi-target angle estimation, which serves as the performance metric for sensing. Subsequently, we formulate the transmit beamforming optimization problem for MiLAC-aided MIMO radar sensing under power constraint, and lossless and reciprocal constraints. In addition, we consider the digital beamforming design as a benchmark for performance comparison.

\subsection{Cram\'{e}r--Rao Bound for 2D Sensing}
Sensing performance is characterized by the CRB, which gives a lower bound on the mean square error (MSE) achievable by any unbiased estimator \cite{Kay_B_1993}. {The CRB is determined by the likelihood of the raw observation $\mathbf{R}$ in \eqref{eq:rx_signal}, not by specific post-processing applied to it \cite{Kay_B_1993}, \eg designing $\mathbf{V}_\mathsf{MiLAC}$ to realize DFT-based DoA detection in our case. Consequently, the CRB depends only on the transmit beamforming matrix $\bW$ and is independent of $\mathbf{V}_\mathsf{MiLAC}$.}
Let $\mathbf{J}$ denote the Fisher information matrix (FIM). Then, the CRB matrix is given by $\boldsymbol{\Xi} = \mathbf{J}^{-1}$.
For $K$ targets, we define the unknown parameter vector as
\begin{equation}
    \boldsymbol{\xi} = [\boldsymbol{\theta}^\top, \boldsymbol{\alpha}^\top]^\top \in \mathbb{R}^{4K \times 1},
\end{equation}
where $\boldsymbol{\theta} \triangleq [\boldsymbol{\theta}_1^\top, \dots, \boldsymbol{\theta}_K^\top]^\top$ collects the target angles with $\boldsymbol{\theta}_k \triangleq [\vartheta_k, \varphi_k]^\top$, and $\boldsymbol{\alpha} \triangleq [\boldsymbol{\alpha}_1^\top, \dots, \boldsymbol{\alpha}_K^\top]^\top$ collects the real and imaginary parts of the reflection coefficients, with $\boldsymbol{\alpha}_k \triangleq [\mathfrak{R}\{\alpha_{r,k}\}, \mathfrak{I}\{\alpha_{r,k}\}]^\top$.

Based on \eqref{eq:rx_signal}, the noise-free received signal is
\begin{equation}
    \begin{aligned}
    \boldsymbol{\mu} & = \sum_{k=1}^K \alpha_{r,k} \mathbf{A}_k(\psi_{x,k}, \psi_{y,k}) \mathbf{F}_\mathsf{MiLAC} \operatorname{diag}(\mathbf{p})^\frac{1}{2} \mathbf{S} \\
    & = \frac{1}{2} \sum_{k=1}^K \alpha_{r,k} \mathbf{A}_k(\psi_{x,k}, \psi_{y,k}) \mathbf{W} \mathbf{S} = \mathbf{A}\bW\mathbf{S},
    \end{aligned}
\end{equation}
where $\mathbf{W} \in \mathcal{W}_\mathsf{MiLAC}$, and the composite round-trip steering matrix is defined as
\begin{equation}
\label{eq:composite_A}
\mathbf{A} \triangleq \frac{1}{2}\sum_{k=1}^K \alpha_{r,k}\,\mathbf{A}_k(\psi_{x,k},\psi_{y,k}) \in \C^{N \times N}.
\end{equation}
The factor of $\frac{1}{2}$ arises because the transmit power constraint is imposed on the voltage-generator signal $\mathbf{C}$, not on the antenna signal $\mathbf{X}$. Specifically, the MiLAC transfer gives $\mathbf{X}=\bW\mathbf{S}/2$ \cite{nerini2025capacity}.\footnote{The same $\frac{1}{2}$ factor appears in the fully-digital case, so the noise-free received signal retains the expression above with $\mathbf{W}$ replaced by the digital beamforming matrix. The set of the digital beamforming matrix is given by $\mathcal{W}_\mathsf{Digital} = \{\mathbf{W}\mathrel{\in}\C^{N\times M} \mid \| \mathbf{W} \|^2_F  \leq P_T\}$, which also imposes the same total transmit power constraint as the MiLAC-aided beamforming. This ensures a fair comparison between MiLAC-aided and fully-digital beamforming under the same transmit power budget.}

Since each parameter $\xi_i$ belongs to a specific target $k$, and only $\mathbf{A}_k$ depends on the parameters of target $k$, the derivatives of $\boldsymbol{\mu} = \mathbf{A}\bW\mathbf{S}$ w.r.t.\ target $k$ are given by
\begin{subequations} \label{eq:76}
\begin{align}
\frac{\partial \boldsymbol{\mu}}{\partial \vartheta_k} &= \frac{\partial \mathbf{A}}{\partial \vartheta_k}\bW \mathbf{S} = \frac{1}{2} \alpha_{r,k} \, \dot{\mathbf{A}}_{\vartheta_k}\mathbf{W} \mathbf{S}, \label{eq:76a} \\
\frac{\partial \boldsymbol{\mu}}{\partial \varphi_k} &= \frac{\partial \mathbf{A}}{\partial \varphi_k}\bW \mathbf{S} = \frac{1}{2} \alpha_{r,k} \, \dot{\mathbf{A}}_{\varphi_k}\mathbf{W} \mathbf{S}, \label{eq:76b} \\
\frac{\partial \boldsymbol{\mu}}{\partial \boldsymbol{\alpha}_k} &= \frac{\partial \mathbf{A}}{\partial \boldsymbol{\alpha}_k}\bW \mathbf{S} = \begin{bmatrix} 1 & \jmath \end{bmatrix}^\top \otimes \frac{1}{2} \mathbf{A}_k\mathbf{W} \mathbf{S}, \label{eq:76c}
\end{align}
\end{subequations}
where $\dot{\mathbf{A}}_{\vartheta_k} \triangleq \frac{\partial \mathbf{A}_k}{\partial \vartheta_k}$ and $\dot{\mathbf{A}}_{\varphi_k} \triangleq \frac{\partial \mathbf{A}_k}{\partial \varphi_k}$ denote the derivatives of the per-target round-trip steering matrix w.r.t.\ the elevation and azimuth angles, respectively. Their explicit expressions are provided in Appendix A.

For estimating $\boldsymbol{\xi}$ from $L$ observations, the FIM has the block structure \cite{Kay_B_1993}
\begin{equation}
\label{eq:fim_mat}
    \mathbf{J}=\left[\begin{array}{ll} \mathbf{J}_{\boldsymbol{\theta} \boldsymbol{\theta}} & \mathbf{J}_{\boldsymbol{\theta} \boldsymbol{\alpha}} \\ \mathbf{J}_{\boldsymbol{\theta} \boldsymbol{\alpha}}^\top & \mathbf{J}_{\boldsymbol{\alpha} \boldsymbol{\alpha}} \end{array}\right].
\end{equation}
The submatrices $\mathbf{J}_{\boldsymbol{\theta} \boldsymbol{\theta}} \in \mathbb{R}^{2K \times 2K}$, $\mathbf{J}_{\boldsymbol{\theta} \boldsymbol{\alpha}} \in \mathbb{R}^{2K \times 2K}$, and $\mathbf{J}_{\boldsymbol{\alpha} \boldsymbol{\alpha}} \in \mathbb{R}^{2K \times 2K}$ are assembled from the contributions of all target pairs. Each scalar entry of the FIM is given by
\begin{equation}
\label{eq:fim_entry}
    [\mathbf{J}]_{i, j} =\frac{2}{\sigma_r^2} \mathfrak{R}\left\{\sum_{\ell=1}^L \frac{\partial \boldsymbol{\mu}[\ell]^H}{\partial \xi_i} \frac{\partial \boldsymbol{\mu}[\ell]}{\partial \xi_j}\right\}, \quad i, j \in \{1, \dots, 4K\}.
\end{equation}
Substituting the derivatives \eqref{eq:76}, which share the common structure $\frac{\partial \boldsymbol{\mu}}{\partial \xi_i} = \frac{\partial \mathbf{A}}{\partial \xi_i}\bW\mathbf{S}$, into \eqref{eq:fim_entry} and using $\mathbf{S}\mathbf{S}^H=L\mathbf{I}$ yields
\begin{equation}
\label{eq:fim_trace}
    [\mathbf{J}]_{i, j} = \frac{2 L}{\sigma_r^2} \Re\left\{\operatorname{tr}\left(\frac{\partial \mathbf{A}^H}{\partial \xi_i} \frac{\partial \mathbf{A}}{\partial \xi_j} \mathbf{R}_x\right)\right\},
\end{equation}
where $\mathbf{R}_x \triangleq \bW\bW^H$ is the transmit covariance matrix. This shows that the FIM, and hence the CRB, depends on $\bW$ only through $\mathbf{R}_x$.

For later use, define $\mathbf{H}_i\triangleq\partial\mathbf{A}/\partial\xi_i\in\C^{N\times N}$. The power-detector receiver measures average powers rather than the complex signals at all receive antennas used in this CRB. For this receiver, the CRB is used as a theoretical lower bound on the MSE of unbiased estimators.

\subsection{Equivalence of MiLAC-aided and Digital Beamforming in Terms of CRB}
\label{subsec:crb_equiv}
The following proposition shows that, despite the restricted beamforming flexibility of lossless and reciprocal MiLAC, the MiLAC-aided and fully-digital designs achieve the same optimal CRB.

\begin{proposition}
\label{prop:same_crb} For any number of waveforms $1\leq M\leq N$, let $\mathcal{R}_\mathsf{MiLAC} \triangleq \{\mathbf{R}_x = \bW\bW^H \mid \bW \in \mathcal{W}_\mathsf{MiLAC}\}$ and $\mathcal{R}_\mathsf{Digital} \triangleq \{\mathbf{R}_x = \bW\bW^H \mid \bW \in \mathcal{W}_\mathsf{Digital}\}$ denote the sets of transmit covariance matrices achievable by MiLAC-aided and fully-digital beamforming, respectively. Then
\begin{equation}
\label{eq:same_cov_set}
\begin{split}
\mathcal{R}_\mathsf{MiLAC} = \mathcal{R}_\mathsf{Digital} = \{\mathbf{R}_x \succeq \mathbf{0} \mid& \operatorname{rank}(\mathbf{R}_x)\leq M,\\
&\trace(\mathbf{R}_x) \leq P_T\}.
\end{split}
\end{equation}
Consequently, MiLAC-aided and fully-digital transmit beamforming designs achieve the same optimal CRB value when constrained to the same number of waveforms.
\end{proposition}

\begin{proof}
We prove the two inclusions separately.

$(\mathcal{R}_\mathsf{MiLAC} \subseteq \mathcal{R}_\mathsf{Digital})$: Since every MiLAC-achievable beamforming matrix is also digitally achievable, \ie $\mathcal{W}_\mathsf{MiLAC} \subseteq \mathcal{W}_\mathsf{Digital}$ \cite{wu2026microwave}, any transmit covariance $\mathbf{R}_x = \bW\bW^H$ realized by some $\bW \in \mathcal{W}_\mathsf{MiLAC}$ is also realized by the same $\bW \in \mathcal{W}_\mathsf{Digital}$.

$(\mathcal{R}_\mathsf{Digital} \subseteq \mathcal{R}_\mathsf{MiLAC})$: Let $\mathbf{R}_x \in \mathcal{R}_\mathsf{Digital}$ be an arbitrary digitally achievable transmit covariance. By definition, there exists some $\bW' \in \mathcal{W}_\mathsf{Digital}$ with $\bW'{\bW'}^H = \mathbf{R}_x$. To show $\mathbf{R}_x \in \mathcal{R}_\mathsf{MiLAC}$, it suffices to find \emph{one} $\bW \in \mathcal{W}_\mathsf{MiLAC}$ satisfying $\bW\bW^H = \mathbf{R}_x$. Since $\bW'\in\C^{N\times M}$, $r\triangleq\operatorname{rank}(\mathbf{R}_x)\leq M$ and $\trace(\mathbf{R}_x)\leq P_T$. For $\mathbf{R}_x=\mathbf{0}$, set $\bW=\mathbf{0}$ and $\mathbf{p}=\mathbf{0}$. Otherwise, consider its compact eigendecomposition $\mathbf{R}_x=\mathbf{U}_r\boldsymbol{\Lambda}_r\mathbf{U}_r^H$, where $\mathbf{U}_r\in\C^{N\times r}$ has orthonormal columns and $\boldsymbol{\Lambda}_r=\operatorname{diag}(\lambda_1,\ldots,\lambda_r)\in\R^{r\times r}$ contains the positive eigenvalues. Construct the following covariance factor
\begin{equation}\label{eq:cov_factor}
\begin{split}
\bW&=[\mathbf{U}_r\boldsymbol{\Lambda}_r^{1/2},\mathbf{0}_{N\times(M-r)}],\\
\mathbf{p}&=(\lambda_1,\ldots,\lambda_r,0,\ldots,0)^\top\in\R_+^{M\times1}.
\end{split}
\end{equation}
Then: (i) $\bW\bW^H=\mathbf{R}_x$; (ii) $\bW^H\bW=\operatorname{diag}(\mathbf{p})$; (iii) $\mathbf{1}^\top\mathbf{p}=\trace(\mathbf{R}_x)\leq P_T$. Hence $\bW\in\mathcal{W}_\mathsf{MiLAC}$ by \eqref{MiLAC3}.

Combining both inclusions yields $\mathcal{R}_\mathsf{MiLAC} = \mathcal{R}_\mathsf{Digital}$. Since the FIM $\mathbf{J}$ in \eqref{eq:fim_mat} is a function of $\mathbf{R}_x = \bW\bW^H$ only, the CRB $c(\mathbf{R}_x)\triangleq\trace(\mathbf{J}(\mathbf{R}_x)^{-1})$ is identical for any two beamforming matrices sharing the same $\mathbf{R}_x$. This property follows from \eqref{eq:fim_trace} for given target parameters and $\mathbf{S}\mathbf{S}^H=L\mathbf{I}$. It does not use the set equality, which appears only in the optimal-value step
\begin{equation}\label{eq:same_opt_crb}
\begin{split}
&\min_{\bW\in\mathcal{W}_\mathsf{MiLAC}}c(\bW\bW^H)=\min_{\mathbf{R}_x\in\mathcal{R}_\mathsf{MiLAC}}c(\mathbf{R}_x)\\
&\qquad=\min_{\mathbf{R}_x\in\mathcal{R}_\mathsf{Digital}}c(\mathbf{R}_x)=\min_{\bW\in\mathcal{W}_\mathsf{Digital}}c(\bW\bW^H).
\end{split}
\end{equation}
The first and last equalities follow from the definitions of $\mathcal{R}_\mathsf{MiLAC}$ and $\mathcal{R}_\mathsf{Digital}$ as the sets of all $\bW\bW^H$ with $\bW$ in $\mathcal{W}_\mathsf{MiLAC}$ and $\mathcal{W}_\mathsf{Digital}$, respectively. The middle equality is the set equality. The same chain holds for the weighted CRB in \eqref{eq:opt}, which also depends on $\bW$ only through $\mathbf{R}_x$. Therefore, the equality of the covariance sets implies that MiLAC-aided and fully-digital transmit beamforming designs achieve the same optimal CRB value.
\end{proof}

\begin{remark}
\label{rmk:crb_vs_sumrate}
Proposition~\ref{prop:same_crb} stands in contrast to the communication setting in \cite{wu2026microwave}, where lossless and reciprocal MiLAC-aided beamforming leads to a sum-rate loss compared to fully-digital beamforming at high signal-to-noise ratio (SNR). The difference arises because the CRB depends on $\bW$ only through $\mathbf{R}_x = \bW\bW^H$. In contrast, the multiuser sum-rate depends on the individual columns of $\bW$ via the per-user signal-to-interference-plus-noise ratio. The spectral-norm constraint $\|\bF\|_2 \leq 1$ imposed by lossless and reciprocal MiLAC restricts the column correlations of $\bW$, which excludes full-power matrices with non-orthogonal columns \cite{wu2026microwave}. This restriction affects the sum-rate, which is sensitive to the column-level structure, but does not affect the CRB.
\end{remark}

\subsection{Transmit Covariance Optimization}
\label{subsec:tx_design}
We aim to design the beamforming matrix $\bW$ to minimize the weighted CRB across all target parameters and thus improve the sensing performance, subject to the MiLAC lossless and reciprocal constraints.
We consider a target tracking scenario. The beamformer is designed for the target count, angles, and reflection coefficients predicted from the estimates of the previous coherent processing interval. Let $\boldsymbol{\omega}=(\omega_1,\ldots,\omega_{4K})^{\top}\in\R_+^{4K\times1}$ contain strictly positive weights. These weights specify the relative scaling of the angle and reflection-coefficient errors. The design problem is
\begin{subequations}\label{eq:opt}
\begin{align}
\min_{\bW,\boldsymbol{\Theta},\mathbf{p}}\quad &
\sum_{i=1}^{4K}\omega_i[\mathbf{J}(\bW\bW^H)^{-1}]_{i,i}\\
\text{s.t.}\quad &\bW=[\boldsymbol{\Theta}]_{M+1:M+N,1:M}\operatorname{diag}(\mathbf{p})^{1/2},\\
&\boldsymbol{\Theta}^{\top}=\boldsymbol{\Theta},\\
&\boldsymbol{\Theta}^H\boldsymbol{\Theta}=\mathbf{I},\\
&\mathbf{p}\geq\mathbf{0},\\
&\mathbf{1}^{\top}\mathbf{p}\leq P_T,
\end{align}
\end{subequations}
By Proposition~\ref{prop:same_crb}, \eqref{eq:opt} is equivalent to
\begin{subequations}\label{eq:opt_rankM}
\begin{align}
\min_{\mathbf{R}_x}\quad &f(\mathbf{R}_x)\triangleq\sum_{i=1}^{4K}\omega_i[\mathbf{J}(\mathbf{R}_x)^{-1}]_{i,i}\label{eq:opt_rankM_obj}\\
\text{s.t.}\quad &\mathbf{R}_x\succeq\mathbf{0},\\
&\trace(\mathbf{R}_x)\leq P_T,\\
&\operatorname{rank}(\mathbf{R}_x)\leq M.\label{eq:opt_rankM_rank}
\end{align}
\end{subequations}
Constraint \eqref{eq:opt_rankM_rank} comes from the $M$ waveforms. Since $\bW\in\C^{N\times M}$ has one column per waveform, every achievable covariance $\mathbf{R}_x=\bW\bW^H$ has rank at most $M$. Without \eqref{eq:opt_rankM_rank}, an optimal covariance could have rank larger than $M$ and could not be generated with $M$ waveforms. For $M<N$, the rank constraint \eqref{eq:opt_rankM_rank} is nonconvex. Dropping it gives a lower bound through the following SDP relaxation.
\begin{subequations}\label{eq:opt_cov_sdp}
\begin{align}
\min_{\mathbf{R}_x,\mathbf{t}}\quad &\sum_{i=1}^{4K}\omega_i t_i\\
\text{s.t.}\quad &\begin{bmatrix}\mathbf{J}(\mathbf{R}_x)&\mathbf{e}_i\\
\mathbf{e}_i^{\top}&t_i\end{bmatrix}\succeq\mathbf{0},\quad i=1,\ldots,4K,\\
&\mathbf{R}_x\succeq\mathbf{0},\\
&\trace(\mathbf{R}_x)\leq P_T,
\end{align}
\end{subequations}
where $\mathbf{t}=(t_1,\ldots,t_{4K})^{\top}\in\R^{4K\times1}$ and $\mathbf{e}_i\in\R^{4K\times1}$ is the $i$th standard basis vector. The Schur complement yields $t_i\geq[\mathbf{J}^{-1}]_{i,i}$. Since the FIM is affine in $\mathbf{R}_x$, \eqref{eq:opt_cov_sdp} is convex.

For $M=N$, the rank constraint \eqref{eq:opt_rankM_rank} always holds. Hence, \eqref{eq:opt_cov_sdp} gives the optimal weighted CRB with $N$ waveforms. Since a MiLAC with $M$ waveforms needs only $M$ RF chains, an important question is how many waveforms are needed to achieve this value.

The FIM depends on the transmit signal only through a few spatial directions. In CRB-based waveform design for MIMO radar with one angle per target and a total power constraint, there exists an optimal transmit covariance of rank at most $2K$ whose columns lie in the subspace spanned by the conjugate transmit steering vectors and their angle derivatives \cite{li2007range}. In our model, write $\mathbf{a}_k=\mathbf{a}(\psi_{x,k},\psi_{y,k})$, $\mathbf{a}_{\vartheta,k}=\partial\mathbf{a}_k/\partial\vartheta_k$, and $\mathbf{a}_{\varphi,k}=\partial\mathbf{a}_k/\partial\varphi_k$, all in $\C^{N\times1}$, and define
\begin{equation}\label{eq:tx_derivative_space}
\mathcal{V}=\operatorname{span}\{\mathbf{a}_k^*,\mathbf{a}_{\vartheta,k}^*,\mathbf{a}_{\varphi,k}^*:k=1,\ldots,K\}.
\end{equation}
Let $d=\dim\mathcal{V}\leq\min(N,3K)$, and let $\mathbf{V}\in\C^{N\times d}$ have orthonormal columns spanning $\mathcal{V}$. The complex conjugates in \eqref{eq:tx_derivative_space} follow from the transpose $\mathbf{a}_k^{\top}$ in the echo model. Every transmit vector $\mathbf{x}\in\C^{N\times1}$ orthogonal to $\mathcal{V}$ gives $\mathbf{H}_i\mathbf{x}=\mathbf{0}$ for all $i$. Such a vector does not illuminate any target, even to first order under small angle changes. Therefore, power sent outside $\mathcal{V}$ adds no Fisher information about $\boldsymbol{\xi}$.
\begin{proposition}\label{prop:stream_bound}
Problem \eqref{eq:opt_cov_sdp} has an optimal covariance of rank at most $d$. Therefore, $M\geq d$ waveforms achieve the same optimal weighted CRB as $M=N$ waveforms. The conservative condition $M\geq\min(N,3K)$ is sufficient.
\end{proposition}
\begin{proof}
Let $\mathbf{P}=\mathbf{V}\mathbf{V}^H\in\C^{N\times N}$ denote the projection onto $\mathcal{V}$, where $\mathbf{V}^H\mathbf{V}=\mathbf{I}_d$. For each parameter $\xi_i$ associated with target $k$, each row of $\mathbf{H}_i$ is a linear combination of $\mathbf{a}_k^{\top}$, $\mathbf{a}_{\vartheta,k}^{\top}$, and $\mathbf{a}_{\varphi,k}^{\top}$. Consequently, $\mathbf{H}_i\mathbf{P}=\mathbf{H}_i$. For any feasible covariance $\mathbf{R}_x$, the projection removes the power sent outside $\mathcal{V}$ and gives $\mathbf{R}_{\mathsf{p}}=\mathbf{P}\mathbf{R}_x\mathbf{P}$. This covariance satisfies
\begin{equation}
\begin{gathered}
\trace(\mathbf{H}_i^H\mathbf{H}_j\mathbf{R}_{\mathsf{p}})=\trace(\mathbf{H}_i^H\mathbf{H}_j\mathbf{R}_x),\\
\trace(\mathbf{R}_{\mathsf{p}})\leq\trace(\mathbf{R}_x),\qquad
\operatorname{rank}(\mathbf{R}_{\mathsf{p}})\leq d.
\end{gathered}
\end{equation}
By \eqref{eq:fim_trace}, the projection keeps the FIM unchanged without increasing the power. Applied to an optimal solution of \eqref{eq:opt_cov_sdp}, it gives an optimal solution of rank at most $d$. For $M\geq d$, this solution satisfies \eqref{eq:opt_rankM_rank}. Hence, it also solves \eqref{eq:opt_rankM} globally.
\end{proof}

We next solve \eqref{eq:opt_rankM} for a given number of waveforms $M$. The projection in the proof of Proposition~\ref{prop:stream_bound} allows the SDP to use $\mathbf{R}_x=\mathbf{V}\mathbf{Q}\mathbf{V}^H$ with $\mathbf{Q}\in\C^{d\times d}$, $\mathbf{Q}\succeq\mathbf{0}$, and $\trace(\mathbf{Q})\leq P_T$. This reduces the matrix dimension from $N$ to $d$ without changing the optimal value. If $M\geq d$, the rank constraint \eqref{eq:opt_rankM_rank} always holds. This is because $\mathbf{V}\mathbf{Q}\mathbf{V}^H$ has rank at most $d$. Hence, solving this SDP gives a globally optimal solution of \eqref{eq:opt_rankM}.

If $M<d$, we first solve the same SDP without the rank constraint \eqref{eq:opt_rankM_rank}. If the SDP optimum has rank at most $M$, it also solves \eqref{eq:opt_rankM} globally. Otherwise, $M$ waveforms cannot realize it. In this case, we solve \eqref{eq:opt_cov_sdp} again with $\mathbf{R}_x=\mathbf{V}_M\mathbf{Q}_M\mathbf{V}_M^H$ and $\mathbf{Q}_M\in\C^{M\times M}$. Here, $\mathbf{V}_M\in\C^{N\times M}$ contains the eigenvectors of the SDP optimum for its $M$ largest eigenvalues. The solution of this reduced SDP has rank at most $M$. It is a feasible solution of \eqref{eq:opt_rankM} without a guarantee of global optimality. Both MiLAC and fully-digital beamforming with $M$ waveforms realize this covariance.

\subsection{Closed-Form MiLAC Realization of the Optimized Covariance}
Given a feasible covariance of rank at most $M$, \eqref{eq:cov_factor} determines its powers. If we extend $\mathbf{U}_r$ to $\mathbf{U}_1\in\C^{N\times M}$ with orthonormal columns, then $\bW=\mathbf{U}_1\operatorname{diag}(\mathbf{p})^{1/2}$ realizes the same covariance. This factorization is also a singular value decomposition because the covariance is Hermitian positive semidefinite. The completion in \cite{wu2026microwave} gives a symmetric and unitary scattering matrix. Extending $\mathbf{U}_1$ to a unitary matrix $\mathbf{U}=[\mathbf{U}_1,\mathbf{U}_2]\in\C^{N\times N}$ gives the following completion.
\begin{proposition}\label{prop:closed_form_theta}
A lossless and reciprocal MiLAC that realizes the prescribed covariance has scattering matrix
\begin{equation}\label{eq:closed_form_theta}
\boldsymbol{\Theta}=
\begin{bmatrix}
\mathbf{0}_{M\times M}&\mathbf{U}_1^{\top}\\
\mathbf{U}_1&-\mathbf{U}_2\mathbf{U}_2^{\top}
\end{bmatrix},
\end{equation}
with powers from \eqref{eq:cov_factor}.
\end{proposition}
\begin{proof}
Symmetry follows from the transpose blocks. The diagonal blocks of $\boldsymbol{\Theta}\boldsymbol{\Theta}^H$ are $\mathbf{U}_1^{\top}\mathbf{U}_1^*=\mathbf{I}$ and $\mathbf{U}_1\mathbf{U}_1^H+\mathbf{U}_2\mathbf{U}_2^{\top}\mathbf{U}_2^*\mathbf{U}_2^H=\mathbf{I}$. The off-diagonal blocks vanish because $\mathbf{U}_1^{\top}\mathbf{U}_2^*=\mathbf{0}$. The transfer block $\mathbf{F}=\mathbf{U}_1$ gives $\bW\bW^H=\mathbf{U}_1\operatorname{diag}(\mathbf{p})\mathbf{U}_1^H=\mathbf{R}_x$.
\end{proof}
Different choices of $\mathbf{U}_2$ give the same covariance. When $M=N$, $\mathbf{U}_2$ is empty and the lower-right block is zero.

To obtain the circuit parameters, invert \eqref{Eq:theta_Y}:
\begin{equation}\label{eq:admittance_inverse}
\mathbf{Y}=Z_0^{-1}(\mathbf{I}-\boldsymbol{\Theta})(\mathbf{I}+\boldsymbol{\Theta})^{-1}.
\end{equation}
This expression is finite when $-1$ is not an eigenvalue of $\boldsymbol{\Theta}$. Similar to \cite{nerini2025capacity}, we assume that $\mathfrak{I}\{\mathbf{U}\}$ is invertible. This assumption excludes the eigenvalue $-1$. Suppose that $\boldsymbol{\Theta}\mathbf{z}=-\mathbf{z}$ with $\mathbf{z}=[\mathbf{x}^{\top},\mathbf{y}^{\top}]^{\top}$. Then $\mathbf{x}=-\mathbf{U}_1^{\top}\mathbf{y}$ and $\mathbf{U}\mathbf{U}^{\top}\mathbf{y}=\mathbf{y}$. Multiplying by $\mathbf{U}^H$ gives $\mathfrak{I}\{\mathbf{U}\}^{\top}\mathbf{y}=\mathbf{0}$. The invertibility of $\mathfrak{I}\{\mathbf{U}\}$ gives $\mathbf{y}=\mathbf{0}$ and $\mathbf{z}=\mathbf{0}$.

\begin{remark}\label{rmk:separated}
Propositions~\ref{prop:same_crb}--\ref{prop:closed_form_theta} hold for any transmit array and also for separated transmit and receive arrays. With $N_r$ receive antennas, target~$k$ has the round-trip steering matrix $\mathbf{A}_k=\mathbf{b}_k\mathbf{a}_k^{\top}\in\C^{N_r\times N}$, where $\mathbf{a}_k$ and $\mathbf{b}_k$ are its transmit and receive steering vectors. If the receiver knows $\mathbf{S}$ and is synchronized with the transmitter, \eqref{eq:fim_trace} holds with $\mathbf{H}_i\in\C^{N_r\times N}$. Each row of every $\mathbf{H}_i$ is still a linear combination of $\mathbf{a}_k^{\top}$, $\mathbf{a}_{\vartheta,k}^{\top}$, and $\mathbf{a}_{\varphi,k}^{\top}$. Hence, the proofs of Propositions~\ref{prop:same_crb}--\ref{prop:closed_form_theta} remain valid.
\end{remark}

\begin{algorithm}[t]
\SetAlgoCaptionLayout{raggedright}
\DontPrintSemicolon
\caption{Transmit beamforming design for MiLAC-aided sensing}
\label{alg:milac_beamforming}
\KwIn{$K$, $\boldsymbol{\xi}$, $N$, $M$, $L$, $P_T$, $\sigma_r^2$, and $\boldsymbol{\omega}$.}
\KwOut{$\mathbf{p}$, $\boldsymbol{\Theta}$, and $\mathbf{Y}$.}
Solve \eqref{eq:opt_cov_sdp} using $\mathbf{R}_x=\mathbf{V}\mathbf{Q}\mathbf{V}^H$.\;
\If{$\operatorname{rank}(\mathbf{R}_x)>M$}{
Solve \eqref{eq:opt_cov_sdp} using $\mathbf{R}_x=\mathbf{V}_M\mathbf{Q}_M\mathbf{V}_M^H$.\;
}
Factor $\mathbf{R}_x=\mathbf{U}_r\boldsymbol{\Lambda}_r\mathbf{U}_r^H$ and form $\mathbf{p}$ in \eqref{eq:cov_factor}.\;
Complete $\mathbf{U}_r$ to a unitary $\mathbf{U}=[\mathbf{U}_1,\mathbf{U}_2]$ with invertible $\mathfrak{I}\{\mathbf{U}\}$ and form $\boldsymbol{\Theta}$ in \eqref{eq:closed_form_theta}.\;
Compute $\mathbf{Y}$ in \eqref{eq:admittance_inverse}.\;
\end{algorithm}

\subsection{Complexity of Transmit Beamforming}
The total computational complexity of Algorithm~\ref{alg:milac_beamforming} is $\mathcal{O}(K^8\ln(1/\epsilon)+K^4N+(N+M)^3)$, where $\epsilon$ is the accuracy of the SDP solution. The first term is the cost of solving \eqref{eq:opt_cov_sdp} with $\mathbf{R}_x=\mathbf{V}\mathbf{Q}\mathbf{V}^H$ by a standard interior-point method \cite{ben2001lectures,wang2014outage}. Since $d\leq3K$, this SDP has $\mathcal{O}(K^2)$ real variables and $4K+1$ linear matrix inequalities of order $\mathcal{O}(K)$. Hence, its cost does not grow with $N$. The reduced SDP for $M<d$ is smaller and does not increase this order. The second term is the cost of forming $\mathbf{V}$ and the FIM coefficients $\mathbf{V}^H\mathbf{H}_i^H\mathbf{H}_j\mathbf{V}$. Since each $\mathbf{H}_i$ has rank at most two, this cost grows linearly with $N$. The third term is the cost of computing \eqref{eq:cov_factor}, \eqref{eq:closed_form_theta}, and \eqref{eq:admittance_inverse}. Algorithm~\ref{alg:milac_beamforming} runs once per design update. During waveform transmission, a digital implementation multiplies an $N\times M$ beamformer by each symbol vector, which requires $\mathcal{O}(NM)$ arithmetic operations per snapshot. The MiLAC performs that linear map through the RF network. Waveform generation and network control remain necessary.

\section{Receiver-Side MiLAC Design for DoA Estimation}
\label{sec:receiver}
In this section, we design the receiver-side MiLAC to implement the 2D-DFT for analog-domain DoA estimation. We first formulate the optimization problem and derive a closed-form solution. Then, we take the snapshots into account and extend the design to implement the fractional 2D-DFT with power averaging over blocks of snapshots for improved DoA estimation performance.
\subsection{Problem Formulation and Closed-form Solution}
The target angles are estimated by applying the 2D-DFT to the received signal in \eqref{eq:rx_milac_output} and locating the spectral peaks. The $(n, n')^\mathrm{th}$ element of the normalized 2D-DFT matrix $\mathbf{F}_\mathsf{DFT} \in \mathbb{C}^{N \times N}$ is given by \cite{richards2005fundamentals}
\begin{equation}
f_{n, n'}=\frac{1}{\sqrt{N}} e^{-\jmath 2 \pi \frac{\left(n_{x}-1\right)\left(n'_x-1\right)}{N_x}} e^{-\jmath 2 \pi \frac{\left(n_{y}-1\right)\left(n'_y-1\right)}{N_y}},
\label{eq:2ddft}
\end{equation}
where $n_{x}$ and $n_{y}$ are defined in \eqref{eq:nxny}. Here, $n$ indexes an output bin and $n'$ indexes an input antenna. Both index pairs use the ordering in \eqref{eq:nxny}. Our goal is to design $\mathbf{V}_\mathsf{MiLAC}$ within the feasible set \eqref{eq:milac_2DDFT} so that it approximates $\mathbf{F}_\mathsf{DFT}$, therefore enabling analog-domain DoA estimation without high-complexity digital processing.

We measure the approximation quality by the Frobenius-norm error
\begin{equation}
\mathcal{L}=\|\beta \mathbf{V}_\mathsf{MiLAC}-\mathbf{F}_\mathsf{DFT}\|_F^2,
\end{equation}
where $\beta$ is a scaling factor for normalization. The optimization problem is
\begin{mini!}|s|[2]
{\mathbf{V}_\mathsf{MiLAC}}{\mathcal{L}=\|\beta \mathbf{V}_\mathsf{MiLAC}-\mathbf{F}_\mathsf{DFT}\|_F^2}
{\label{eq:opt_2dft}}{\mathcal{P}4:}
\addConstraint{\mathbf{V}_\mathsf{MiLAC} = \frac{1}{2}[\boldsymbol{\Theta}]_{N+1: 2N, 1: N},}
\addConstraint{\boldsymbol{\Theta}=\boldsymbol{\Theta}^\top,}
\addConstraint{\boldsymbol{\Theta}^H \boldsymbol{\Theta}=\mathbf{I}.}
\end{mini!}

Since $\mathbf{F}_\mathsf{DFT}$ is unitary, \ie $\mathbf{F}_\mathsf{DFT}^H \mathbf{F}_\mathsf{DFT} = \mathbf{I}$, the global optimum is obtained at $\mathbf{V}_\mathsf{MiLAC} = \frac{1}{2} \mathbf{F}_\mathsf{DFT}$ and $\beta = 2$ with zero approximation error. Substituting this into \eqref{eq:transfer_matrix}, the corresponding scattering matrix is
\begin{equation}
\label{eq:fim_all}
\boldsymbol{\Theta}=\left[\begin{array}{ll}
\mathbf{0} & \mathbf{F}_\mathsf{DFT}^\top \\
\mathbf{F}_\mathsf{DFT} & \mathbf{0}
\end{array}\right].
\end{equation}
This closed-form solution shows that the receiver-side MiLAC can implement the 2D-DFT exactly, without any offline iterative optimization. The matrix is symmetric, and its product with its Hermitian transpose has diagonal blocks $\mathbf{F}_{\mathsf{DFT}}^{\top}\mathbf{F}_{\mathsf{DFT}}^*=\mathbf{I}$ and $\mathbf{F}_{\mathsf{DFT}}\mathbf{F}_{\mathsf{DFT}}^H=\mathbf{I}$. Its eigenvalues belong to $\{1,-1,\jmath,-\jmath\}$. To obtain a finite admittance in \eqref{eq:admittance_inverse}, we multiply \eqref{eq:fim_all} by the uniform phase $e^{\jmath\phi}$. This phase preserves the symmetry, the unitarity, and every output power. We select $\phi$ to maximize the smallest singular value of $\mathbf{I}+e^{\jmath\phi}\boldsymbol{\Theta}$. Since $\boldsymbol{\Theta}$ is unitary, this value equals the smallest distance between $-1$ and the eigenvalues of $e^{\jmath\phi}\boldsymbol{\Theta}$. Choosing $\phi=\pi/4$ attains this maximum and excludes the eigenvalue $-1$. For any unitary fractional DFT, the phase is selected separately by the same rule.

\subsection{Fractional 2D-DFT for DoA Estimation}
\label{subsec:fracdft}
To improve the DoA estimation performance, we can leverage the snapshots across time slots to implement a fractional 2D-DFT, which provides finer spatial-frequency sampling than the standard 2D-DFT.
The standard DFT samples the spatial spectrum at $N_x\times N_y$ points. The receiver-side MiLAC is set to $G=L_xL_y$ configurations in turn. They shift these points by $L_x$ and $L_y$ fractions of their spacing along the two array axes. Thus, they sample the intervals between the points without additional output ports. The finer grid does not increase the array aperture or ensure target separation.

Each configuration is held fixed for a block of $T_b\geq M$ snapshots, and the total snapshot count is $L=GT_b$. In block $g$, the $M$ RF chains send the waveforms $\mathbf{S}_g\in\C^{M\times T_b}$ with $\mathbf{S}_g\mathbf{S}_g^H=T_b\mathbf{I}$. For $g=1,\ldots,G$, define
\begin{equation}\label{eq:lxy}
g_y=\left\lceil g/L_x\right\rceil,\qquad g_x=g-(g_y-1)L_x.
\end{equation}
The fractional transform is $\mathbf{F}_g=\mathbf{F}_{\mathsf{DFT}}\mathbf{D}_g$, where $\mathbf{D}_g\in\C^{N\times N}$ is diagonal with
\begin{equation}\label{eq:input_phase}
\begin{split}
[\mathbf{D}_g]_{n',n'}={}&\exp\!\left[-\jmath2\pi
\frac{(g_x-1)(n'_x-1)}{N_xL_x}\right]\\
&\times\exp\!\left[-\jmath2\pi
\frac{(g_y-1)(n'_y-1)}{N_yL_y}\right].
\end{split}
\end{equation}
Therefore, the entries of $\mathbf{F}_g$ are given by
\begin{equation}\label{eq:2ddft_L}
\begin{split}
[\mathbf{F}_g]_{n,n'}={}&\frac{1}{\sqrt N}\exp\!\left[-\jmath\frac{2\pi}{N_x}
\left(n_x-1+\frac{g_x-1}{L_x}\right)(n'_x-1)\right]\\
&\times\exp\!\left[-\jmath\frac{2\pi}{N_y}
\left(n_y-1+\frac{g_y-1}{L_y}\right)(n'_y-1)\right].
\end{split}
\end{equation}
The phase shifts multiply the input antenna coordinates. Output-side phase shifts would leave each output power unchanged. Since $\mathbf{F}_g$ is unitary, the construction in \eqref{eq:fim_all} applies to each configuration.

With the echo block $\mathbf{R}_g\in\C^{N\times T_b}$ of configuration $g$, the average power at output $n$ is
\begin{equation}\label{eq:block_power}
P_{n,g}=\frac{1}{T_b}\sum_{t=1}^{T_b}|[\mathbf{F}_g\mathbf{R}_g]_{n,t}|^2.
\end{equation}
This definition removes the factor $\tfrac12$ in \eqref{eq:milac_2DDFT}. This factor scales the signal and the noise equally and does not change the peak locations. A digital implementation computing \eqref{eq:block_power} produces identical powers under the ideal model and the same noise realization. With $\mathbf{f}_{n,g}=[\mathbf{F}_g]_{n,:}$, the orthogonality of $\mathbf{S}_g$ yields the expected output power
\begin{equation}\label{eq:block_power_mean}
\mathbb{E}[P_{n,g}]=
\mathbf{f}_{n,g}\mathbf{A}\mathbf{R}_x\mathbf{A}^H\mathbf{f}_{n,g}^H+\sigma_r^2.
\end{equation}

For each bin and configuration, its electrical angles are
\begin{equation}\label{eq:electrical_recovery}
\begin{split}
\psi_{x,n,g}&=\pi\!\left[\operatorname{mod}\!\left(2\frac{n_x-1+(g_x-1)/L_x}{N_x}+1,2\right)-1\right],\\
\psi_{y,n,g}&=\pi\!\left[\operatorname{mod}\!\left(2\frac{n_y-1+(g_y-1)/L_y}{N_y}+1,2\right)-1\right].
\end{split}
\end{equation}
The $N$ powers of all $G$ configurations form an $(N_xL_x)\times(N_yL_y)$ grid ordered by these electrical angles. For the half-wavelength arrays in Section~\ref{sec:simu}, the normalized electrical angles are $u=\psi_x/\pi=\sin\vartheta\cos\varphi$ and $v=\psi_y/\pi=\sin\vartheta\sin\varphi$. A grid point is valid only if
\begin{equation}\label{eq:visible_region}
\left(\frac{\psi_{x,n,g}}{\kappa d_x}\right)^2+
\left(\frac{\psi_{y,n,g}}{\kappa d_y}\right)^2\leq1.
\end{equation}
The search takes the $K$ largest distinct local maxima among the valid points. A local maximum is no smaller than any of its valid eight neighbors, with periodic indexing in $\psi_x$ and $\psi_y$. A connected plateau counts as one maximum. With the set $\mathcal{M}$ of valid local maxima, the $k$th selected pair is
\begin{equation}\label{eq:peak_search}
(\widehat n_k,\widehat g_k)=\arg\max_{(n,g)\in\mathcal{M}\setminus\widehat{\mathcal{I}}_{k-1}}P_{n,g},\quad k=1,\ldots,K,
\end{equation}
where $\widehat{\mathcal{I}}_{k-1}$ collects the first $k-1$ selected pairs. Unlike selecting the $K$ largest powers, this search does not return adjacent samples of the same peak. For both receivers, the targets are unresolved if fewer than $K$ maxima exist.

Each selected maximum is refined along each electrical angle by a quadratic fit to the log powers $\ell_-$, $\ell_0$, and $\ell_+$ of three adjacent points. The offset in grid intervals is
\begin{equation}\label{eq:peak_refinement}
\delta=\operatorname{clip}_{[-1/2,1/2]}\!\left(
\frac{\ell_--\ell_+}{2(\ell_--2\ell_0+\ell_+)}\right),
\end{equation}
where the clipping operation limits its argument to $[-1/2,1/2]$. The offset is set to zero if the fitted quadratic is not concave or either adjacent point is outside the visible region. If the refined direction violates \eqref{eq:visible_region}, both electrical-angle offsets $2\pi\delta_x/(N_xL_x)$ and $2\pi\delta_y/(N_yL_y)$ are halved until \eqref{eq:visible_region} holds.

For a selected pair $(\widehat n_k,\widehat g_k)$, evaluate \eqref{eq:electrical_recovery} and apply the refinement above to obtain $\widehat\psi_{x,k}$ and $\widehat\psi_{y,k}$. The physical-angle estimates are
\begin{equation}\label{eq:physical_recovery}
\begin{split}
\widehat\varphi_k&=\operatorname{mod}\!\left(\operatorname{atan2}\!\left(
\widehat\psi_{y,k}/d_y,\widehat\psi_{x,k}/d_x\right),2\pi\right),\\
\widehat\vartheta_k&=\arcsin\!\left(\frac{1}{\kappa}
\sqrt{\widehat\psi_{x,k}^2/d_x^2+\widehat\psi_{y,k}^2/d_y^2}\right).
\end{split}
\end{equation}
The factor $\pi$ in \eqref{eq:electrical_recovery} converts normalized spatial frequency into electrical angle. The two-argument inverse tangent preserves the azimuth quadrant. At broadside, azimuth is undefined. At elevation $\pi/2$, the elevation derivative vanishes in this parametrization. These singular directions are excluded from the CRB comparisons.

\section{Numerical Evaluation}
\label{sec:simu}
In this section, we evaluate the performance of the proposed MiLAC-aided MIMO radar sensing system. The system employs a co-located UPA with half-wavelength element spacing, \ie $d_x = d_y = \lambda/2$. The UPA is assumed to be square with $N_x = N_y$ and hence $N = N_x \times N_y$. The noise power is normalized to $\sigma_r^2 = 1$, and the reflection coefficients are $\alpha_{r,k}=0.8e^{\jmath\pi/4}$. The normalized transmit power is $P_{\mathsf{dB}}=10\log_{10}(P_T/\sigma_r^2)$ in dB. Received SNR additionally depends on target attenuation and array gain. To improve the DoA estimation performance, we consider $L$ snapshot observations at the receiver using a fractional DFT-based approach as explained in Section \ref{subsec:fracdft}.
The weighted CRB uses $\omega_i=1$ for the $2K$ angle entries and $\omega_i=10^4$ for the $2K$ real and imaginary reflection-coefficient entries. Angles are measured in radians. The covariance SDP is convex and can be solved efficiently without FIM regularization. Three baseline schemes are considered for comparison: on the transmit side, (i) the \emph{digital beamforming design} from \eqref{eq:opt_cov_sdp}, whose value is achievable when its covariance rank does not exceed the number of waveforms, and (ii) the \emph{matched filter}, which steers the transmit beam toward the target directions using the steering vectors; on the receive side, (iii) the \emph{digital 2D-DFT-based DoA estimator}, which performs the 2D-DFT computation~digitally. The matched filter uses one waveform with the beamformer $\mathbf{w}_{\mathsf{MF}}=\sqrt{P_T}\sum_k\mathbf{a}_k^*/\|\sum_k\mathbf{a}_k^*\|_2\in\C^{N\times1}$.

The results in Section~\ref{subsec:tx_eval} use $L=400$ snapshots. Those at a fixed transmit power use $P_{\mathsf{dB}}=20$. The results in Section~\ref{subsec:rx_eval} use $M=3K$ waveforms for the optimized covariance and $M=N$ waveforms for the omnidirectional baseline. The beampattern, detection, and MSE results use one target at $(\vartheta,\varphi)=(30^\circ,45^\circ)$ or two targets at $(30^\circ,45^\circ)$ and $(60^\circ,90^\circ)$. The CRB results use the first $K$ directions from $(\vartheta,\varphi)=(30^\circ,30^\circ)$, $(45^\circ,60^\circ)$, $(60^\circ,90^\circ)$, $(45^\circ,120^\circ)$, and $(30^\circ,150^\circ)$.

\begin{figure}[t]
    \centering
    \includegraphics[width=0.9\linewidth]{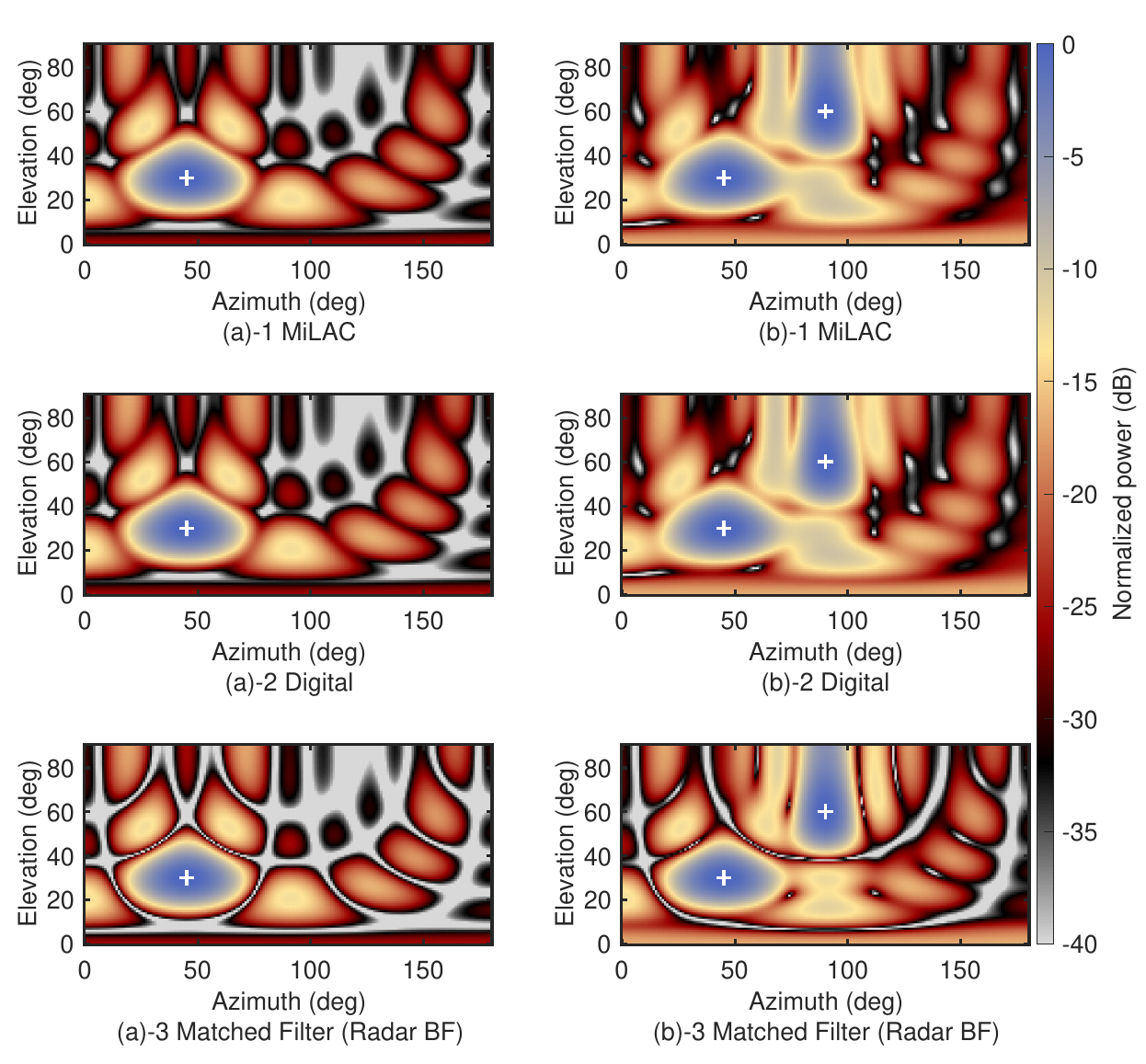}
    \centering
    \caption{Transmit beampatterns with $N=64$ at $P_{\mathsf{dB}}=20$. Columns (a) and (b) show one and two targets with $M=3$ and $M=6$ waveforms, respectively. Rows show the MiLAC, digital, and matched-filter beampatterns.}
    \label{fig:all_beam}
\end{figure}

\begin{table}[t]
\caption{3 dB main-lobe widths and PSL of the beampatterns in \fig{fig:all_beam}.}
\label{tab:beam_metrics}
\centering\small
\begin{tabular}{@{}lccc@{}}
\toprule
Design & \multicolumn{2}{c}{3 dB width ($^\circ$)} & PSL (dB)\\
\cmidrule(lr){2-3}
 & Elevation & Azimuth & \\
\midrule
\multicolumn{4}{@{}l}{One target at $(30^\circ,45^\circ)$}\\
MiLAC & $15.5$ & $26.0$ & $-12.8$\\
Digital & $15.5$ & $26.0$ & $-12.8$\\
Matched filter & $15.1$ & $26.0$ & $-12.8$\\
\midrule
\multicolumn{4}{@{}l}{Two targets at $(30^\circ,45^\circ)$ and $(60^\circ,90^\circ)$}\\
MiLAC & $16.3$, $31.6$ & $26.5$, $15.3$ & $-9.8$\\
Digital & $16.3$, $31.6$ & $26.5$, $15.3$ & $-9.8$\\
Matched filter & $15.0$, $29.0$ & $26.1$, $14.7$ & $-11.6$\\
\bottomrule
\end{tabular}
\end{table}

\subsection{MiLAC-aided Transmit Beamforming Design}
\label{subsec:tx_eval}
\subsubsection{Beampattern Performances}
\fig{fig:all_beam} compares the 2D transmit beampatterns generated by three kinds of beamforming designs with $N = 64$ antennas. It can be observed that: (i) the MiLAC-aided beamforming design can probe to the correct target directions. (ii) The beampattern of the MiLAC design has the same performance as the fully-digital design. The transmit beampattern in direction $(\vartheta,\varphi)$ is $\mathbf{a}^{\top}\mathbf{R}_x\mathbf{a}^*$, where $\mathbf{a}$ is the transmit steering vector at the electrical angles in \eqref{eq:psi_xy}. This beampattern depends on $\bW$ only through $\mathbf{R}_x$. The closed-form MiLAC realization reproduces the SDP covariance of the fully-digital design. Therefore, the two beampatterns are identical.

To quantify the beampattern shapes, we measure the 3 dB main-lobe widths along the elevation and azimuth directions through each main-lobe peak. Each main lobe extends from its peak to the surrounding beampattern minima. The peak sidelobe level (PSL) is the highest beampattern value in the visible region outside the main lobes, given in dB relative to the global peak. All main-lobe peaks lie within $0.3^\circ$ of the target directions. Table~\ref{tab:beam_metrics} lists the measured values. The MiLAC and fully-digital designs have identical widths and PSL. Compared with the matched filter, their main lobes are at most $2.6^\circ$ wider. Their PSL is equal for one target and $1.8$ dB higher for two targets. These differences follow from the CRB objective. The matched filter places all its power on the steering vectors $\mathbf{a}_k^*$. In contrast, the CRB-optimal covariance places part of its power in the derivative directions of \eqref{eq:tx_derivative_space}. This power improves the angle information. For one target, it also makes the amplitude non-uniform across the array, which widens the main lobe compared with the uniform amplitude of the matched filter.

\subsubsection{CRB Versus Transmit Power}
To quantitatively evaluate the estimation performance, we compute the CRB using the FIM derived in \eqref{eq:fim_mat} with the optimal beamformers obtained from Algorithm~\ref{alg:milac_beamforming} and \eqref{eq:opt_cov_sdp}, respectively, and compare the proposed MiLAC-aided design with the fully-digital benchmark.

\begin{figure}[t]
    \centering
    \includegraphics[width=0.75\linewidth]{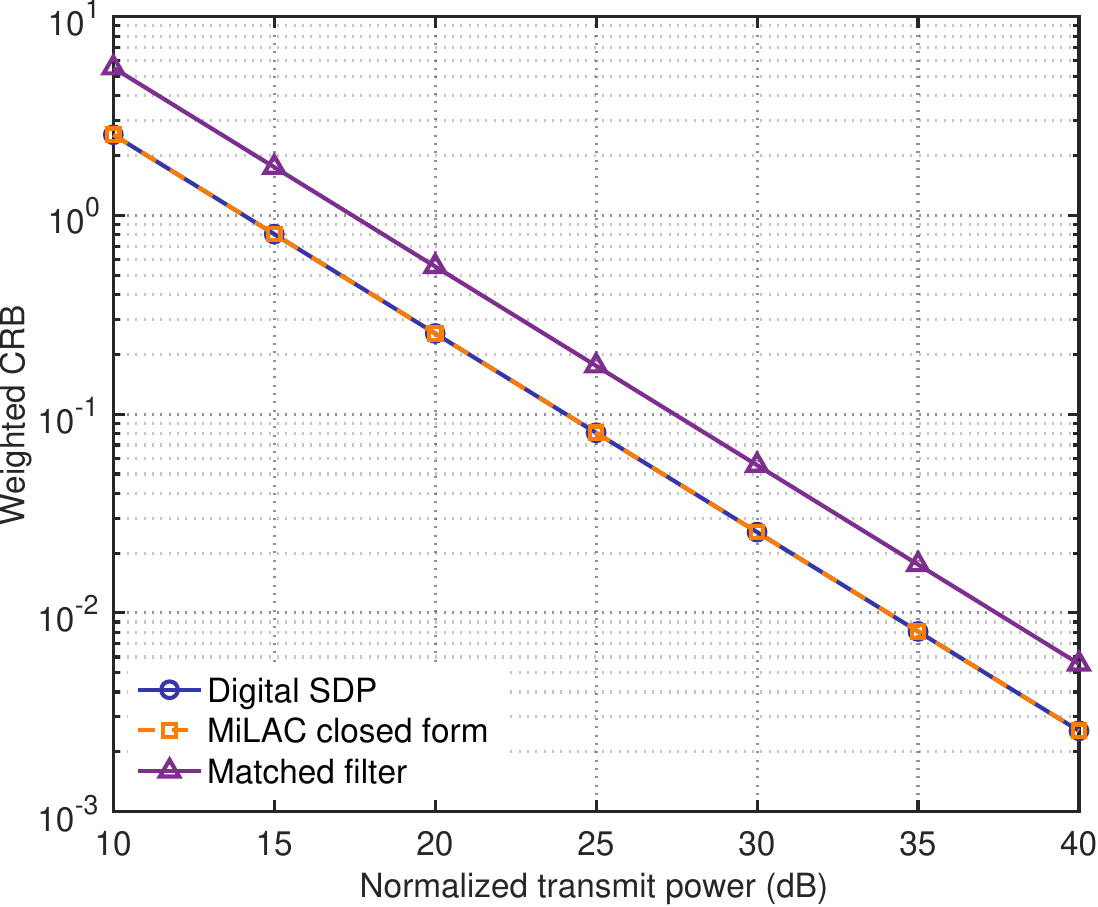}
    \centering
    \caption{Weighted CRB versus normalized transmit power with $K=3$, $N=16$, $M=3K$, and $L=400$.}
    \label{fig:crb_snr_3target}
\end{figure}

We first evaluate the CRB versus normalized transmit power with $K = 3$ targets and $N = 16$ antennas. As shown in \fig{fig:crb_snr_3target}, the CRB decreases with increasing transmit power for both schemes, and the MiLAC design achieves CRB performance the same as the fully-digital design across the entire transmit-power range.

\begin{figure}[t]
    \centering
    \includegraphics[width=0.75\linewidth]{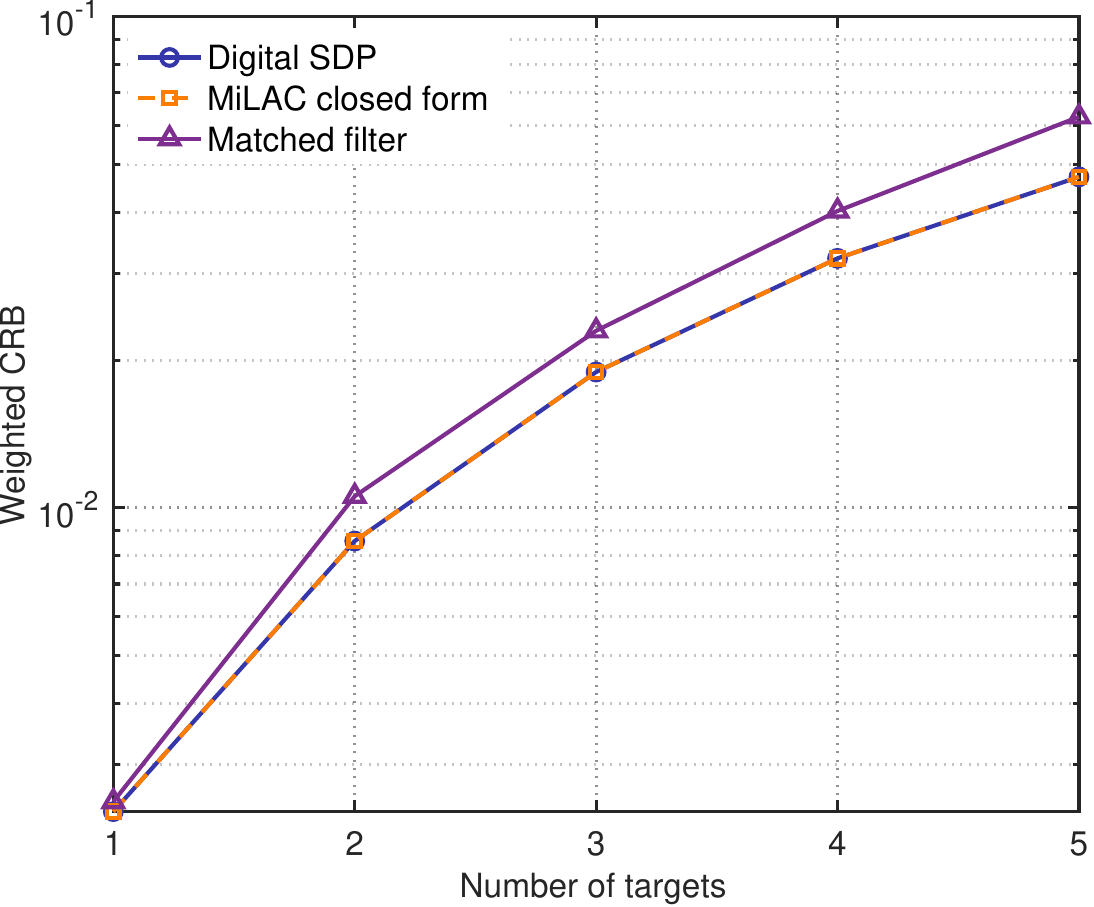}
    \centering
    \caption{Weighted CRB versus target count $K$ with $N=64$, $M=3K$, $L=400$, and $P_{\mathsf{dB}}=20$.}
    \label{fig:crb_target}
\end{figure}

\subsubsection{CRB Versus Number of Targets}
\fig{fig:crb_target} plots the CRB versus the number of targets $K$ with $N  = 64$ antennas. As the number of targets increases, the CRB increases for both designs because more unknown parameters need to be estimated and the transmit power per target is reduced. The MiLAC design consistently achieves CRB performance the same as the fully-digital benchmark. This verifies the theoretical analysis that the MiLAC-aided beamforming design can achieve the same CRB as the fully-digital design in Section~\ref{subsec:crb_equiv}.

\begin{figure}[!t]

\centering
\includegraphics[width=0.75\linewidth]{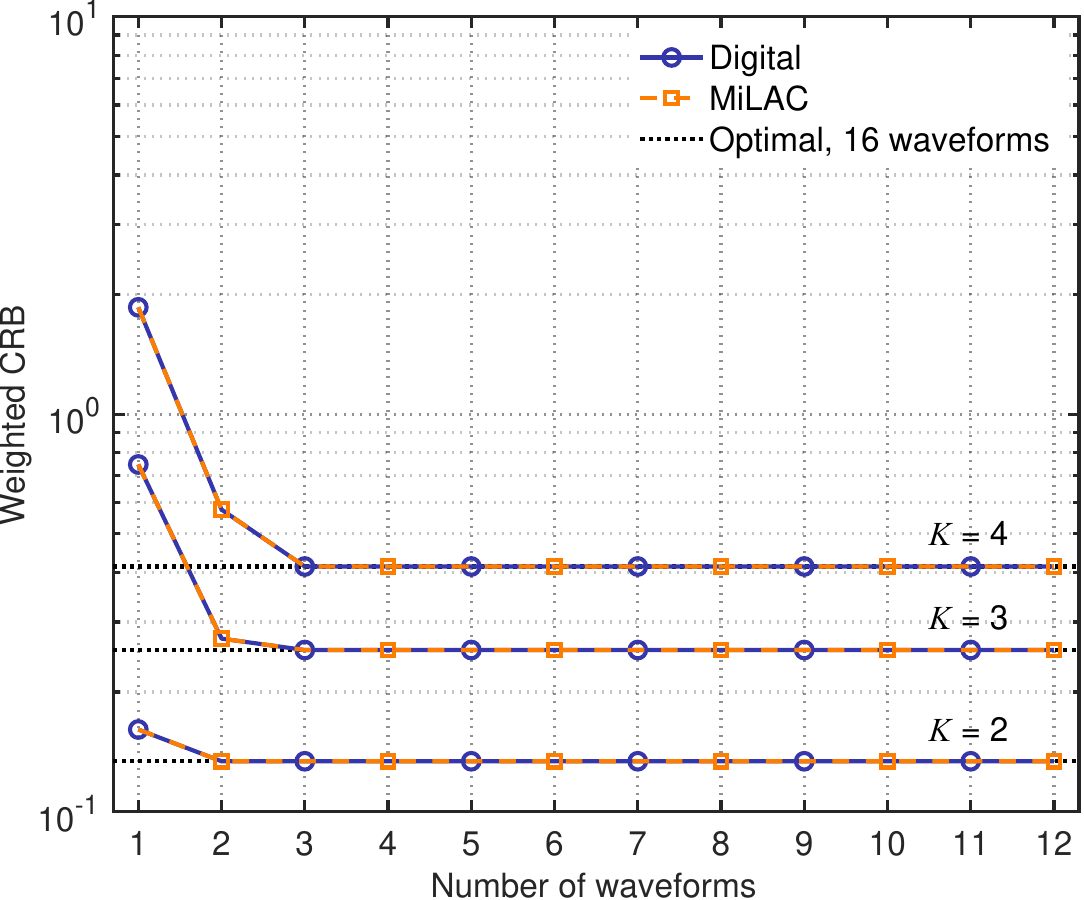}
\caption{Weighted CRB versus the number of waveforms for $N=16$ and $K=2,3,4$.}
\label{fig:crb_stream}
\end{figure}

\begin{figure}[!t]
\centering
\includegraphics[width=0.75\linewidth]{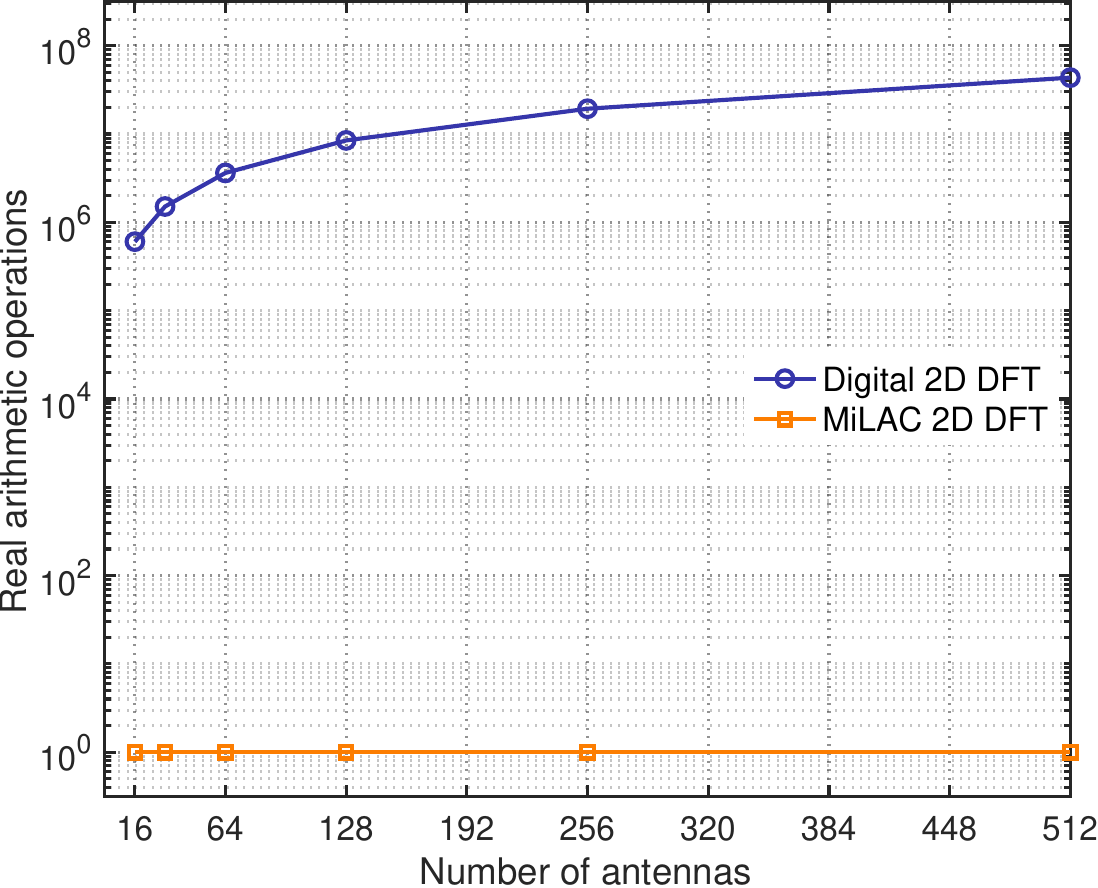}
\caption{Computational complexity of digital and MiLAC 2D-DFT with $L = 2500$. The MiLAC 2D-DFT uses no digital arithmetic, and its zero count is plotted at one.}
\label{fig:2ddft_complexity}
\end{figure}

\begin{figure}[!t]
    \centering
    \includegraphics[width=0.75\linewidth]{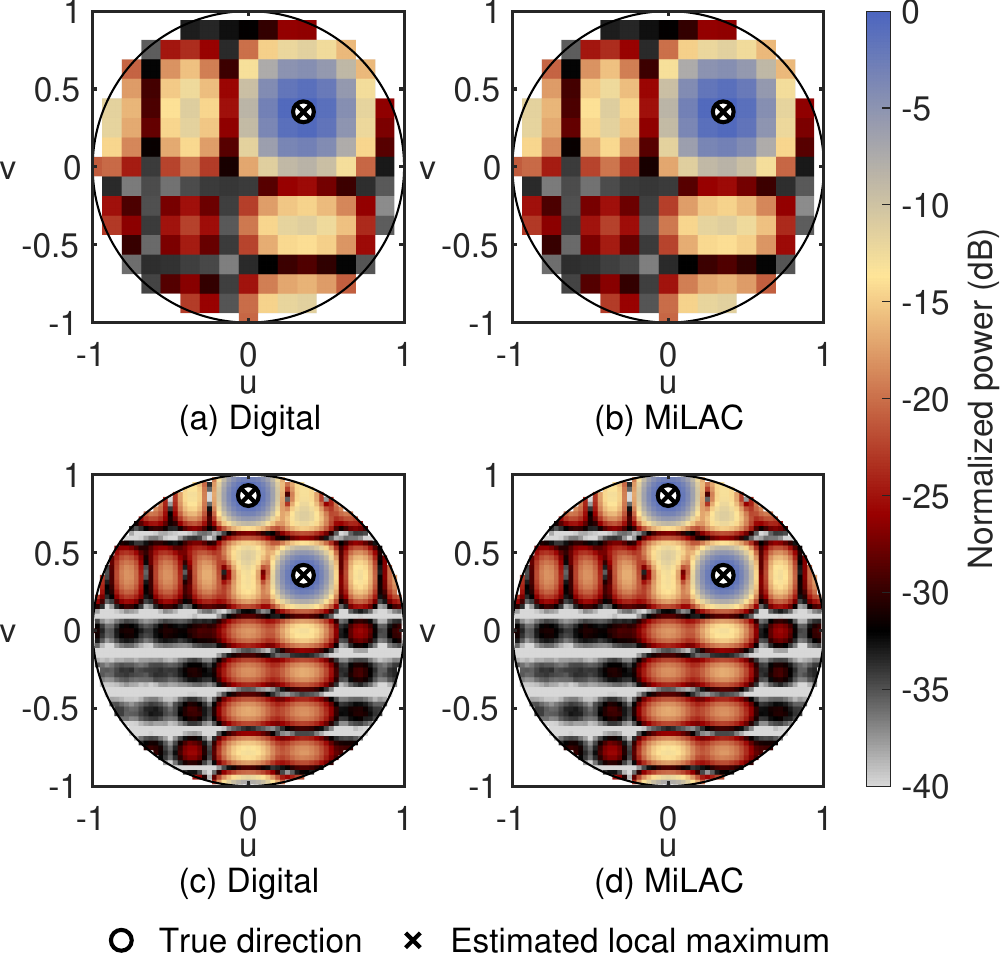}
    \centering
    \caption{Angular power spectra with digital and MiLAC reception at $M=3K$ and $P_{\mathsf{dB}}=20$. (a), (b) One target with $N=16$ and $L=256$. (c), (d) Two targets with $N=64$ and $L=4096$.}
    \label{fig:2ddft}
\end{figure}

\begin{table}[t]

\caption{Receiver operations for $N=N_xN_y$ antennas, $G$ configurations, and $L=GT_b$ snapshots.}
\label{tab:rx_complexity}
\centering\footnotesize
\setlength{\tabcolsep}{4pt}

\begin{tabularx}{\linewidth}{@{}>{\raggedright\arraybackslash}p{0.232\linewidth}>{\raggedright\arraybackslash}p{0.337\linewidth}>{\raggedright\arraybackslash}X@{}}
\toprule
Operation & Digital fractional DFT & MiLAC fractional DFT\\
\midrule
Complex-signal sampling & $N$ RF chains and ADCs & None\\
\midrule
2D-DFT and input phase shifts & $\mathcal{O}(LN\log N)$ FFT and $\mathcal{O}(LN)$ phase operations & RF propagation\\
\midrule
Power averaging & $\mathcal{O}(LN)$ operations & $N$ detectors and integrators\\
\midrule
Peak processing & $\mathcal{O}(NG)$ neighbor comparisons and peak ranking & Same as digital\\
\midrule
Configuration & Digital phase coefficients & $G$ analog configurations\\
\bottomrule
\end{tabularx}
\end{table}

\subsubsection{CRB Versus Number of Waveforms}
This experiment examines when fewer RF chains achieve the optimal weighted CRB with $N$ waveforms. The covariance equality holds at each fixed $M$. However, attaining this optimum additionally requires enough waveforms to realize an optimal covariance. We use $N=16$ and $K=2,3,4$.

\fig{fig:crb_stream} distinguishes the feasible reduced-subspace objectives from the optimal weighted CRB with $N$ waveforms. Within numerical tolerance, this optimum is attained at $M=2,3,3$ for $K=2,3,4$, respectively. The digital and MiLAC objectives agree up to numerical precision across this sweep. These examples require fewer waveforms than the sufficient bound $d=3K$ in Proposition~\ref{prop:stream_bound}. This bound is a worst-case guarantee. In contrast, the number of waveforms needed equals the smallest rank of an optimal covariance, \ie $2,3,3$. This rank depends on the target directions and the weights $\omega_i$.

\subsection{MiLAC-aided Receiver Side DoA Estimator}
\label{subsec:rx_eval}

\begin{figure}[t]
    \centering
    \includegraphics[width=0.75\linewidth]{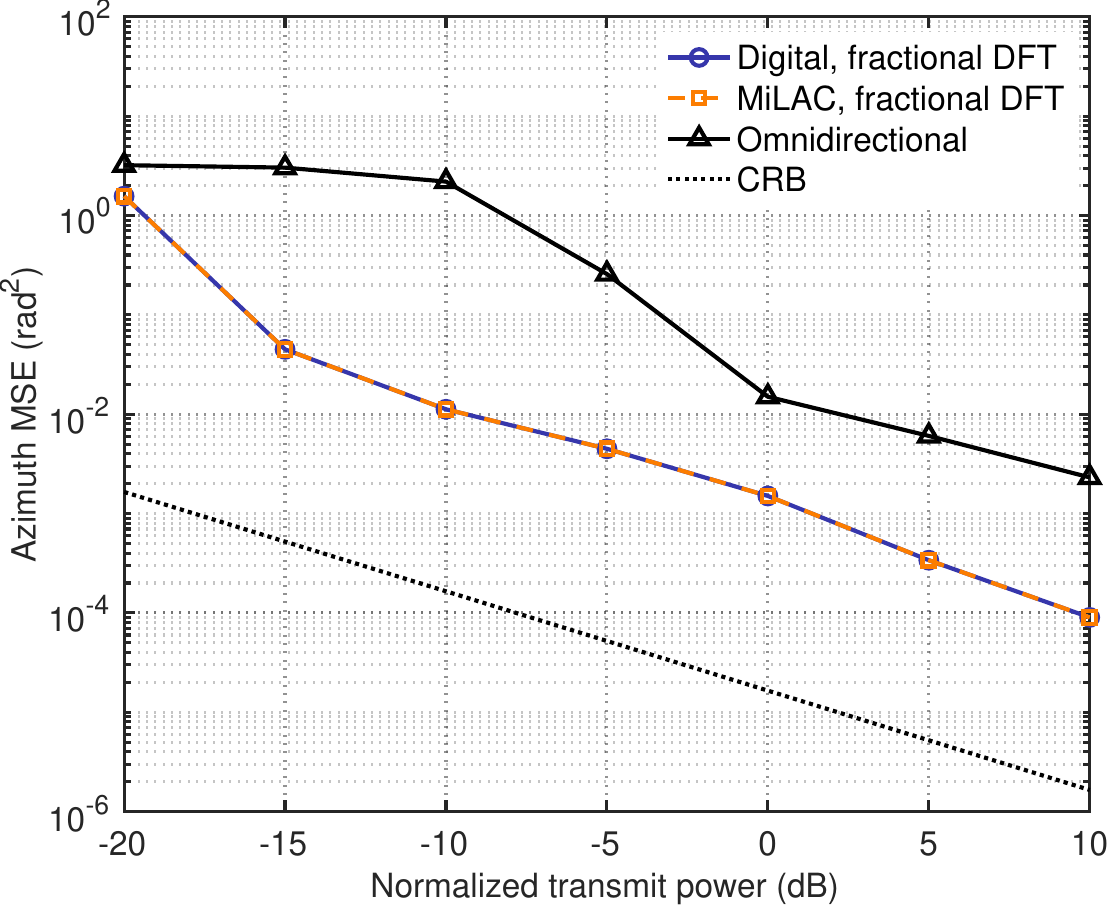}
    \centering
    \caption{Azimuth MSE in $\mathrm{rad}^2$ versus normalized transmit power with $N=16$, $M=3K$ ($M=N$ for the omnidirectional baseline), $L=256$, and 10,000 simulations per power.}
    \label{fig:MSE_ant16}
\end{figure}

\subsubsection{Computational Complexity}
A split-radix two-dimensional fast Fourier transform (2D FFT) on an $N_x\times N_y$ array with power-of-two $N_x$ and $N_y$ uses about $(34/9)N\log_2N$ real arithmetic operations per snapshot \cite{johnson2006modified}, where $N=N_xN_y$. Over $L=GT_b$ snapshots, the digital 2D-DFT uses $(34/9)LN\log_2N$ operations. The other digital operations in Table~\ref{tab:rx_complexity} grow as $LN$ or $NG$, where $G\leq L$. Therefore, the 2D-DFT is the most computation-intensive digital operation. The MiLAC performs this transform and the input phase shifts through RF propagation without digital arithmetic. Spectrum assembly and peak search remain digital after the scalar powers are read out. \fig{fig:2ddft_complexity} compares the operation counts for $L=2500$ and $N_y\in\{N_x,2N_x\}$. The digital count increases with the number of antennas. This comparison counts digital arithmetic only.

\subsubsection{Hardware Complexity}
The fully connected $2N$-port receive network connects each pair of ports and each port to ground through a reconfigurable impedance. Therefore, it has $(2N)(2N+1)/2$ reconfigurable impedances. Each configuration change can re-tune all of them. A lower-complexity realization places $N$ tunable input phase shifters before a fixed DFT network of hybrid couplers and phase shifters \cite{nerini2026analog}. Each configuration change updates only the $N$ input phase shifts in \eqref{eq:input_phase}. Since an $n$-point DFT network uses $(n/2)\log_2n$ couplers \cite{nerini2026analog}, the networks along both array axes use $(N/2)\log_2N$ couplers for power-of-two $N_x$ and $N_y$. In total, this realization has $N$ tunable phase shifters and $\mathcal{O}(N\log N)$ fixed couplers and phase shifters. A fabricated $4\times4$ DFT network has about 1~dB of additional insertion loss and matches the DFT with a normalized MSE of $0.091\%$ after calibration \cite{nerini2026analog}.

\subsubsection{Detection results}
In this part, we evaluate the DoA estimation performance of the MiLAC-aided receiver, which performs the 2D-DFT in the analog domain via the MiLAC \eqref{eq:fim_all}. \fig{fig:2ddft}(a) shows the angular spectrum obtained from the digital receiver for a single target. The spectrum shows a peak at the true target location. Compared to \fig{fig:2ddft}(b) which is the MiLAC-aided approach, this confirms that the MiLAC successfully implements the 2D-DFT and enables accurate DoA detection using analog spatial transforms and scalar-power readout.

We extend the detection to the multi-target case with two targets. \fig{fig:2ddft}(c) shows the digital spectrum, and \fig{fig:2ddft}(d) shows the MiLAC spectrum. Both targets are clearly detected.

\subsubsection{MSE Performances of Parameter Estimation}
Finally, we compare the MSE of the DoA estimation between the MiLAC-aided approach and the fully-digital 2D-DFT-based method. The receiver uses $T_b=N$ snapshots per configuration, with $L_x=L_y=4$ for $N=16$ and $L_x=L_y=8$ for $N=64$. The MSE results use one target with the same direction in all Monte Carlo simulations. The same noise is used for the digital and MiLAC comparisons. The azimuth MSE is $N_{\mathsf{MC}}^{-1}\sum_{r=1}^{N_{\mathsf{MC}}}[\operatorname{wrap}_{(-\pi,\pi]}(\widehat\varphi_r-\varphi)]^2$, where $N_{\mathsf{MC}}$ is the number of simulations and $r$ indexes them. The wrap operation reduces an angle to $(-\pi,\pi]$. The ``MiLAC'' label refers to utilizing the MiLAC-aided transmit beamforming (Algorithm~\ref{alg:milac_beamforming}) combined with MiLAC-aided analog 2D-DFT reception. ``Digital'' refers to the digital transmit beamforming design from \eqref{eq:opt_cov_sdp} combined with digital 2D-DFT processing. The omnidirectional baseline uses $\mathbf{R}_x=P_T\mathbf{I}/N$ with the same fractional-DFT acquisition and peak estimator.

We take the estimation of the azimuth angle $\varphi$ as an example. \fig{fig:MSE_ant16} shows the MSE versus normalized transmit power for a single target. The MiLAC-aided estimator achieves the same MSE as the fully-digital approach across the tested transmit powers. The CRB of the complex signals at all receive antennas is shown as the lower bound on the MSE of any unbiased estimator. Each estimator achieves higher MSE than this bound at all tested powers.

\section{Conclusion}
\label{sec:con}
In this paper, we have proposed a MiLAC-aided MIMO radar sensing approach that implements both transmit beamforming and receiver-side 2D-DFT-based DoA estimation. We have proved that MiLAC-aided and fully-digital beamforming achieve the same optimal CRB with the same number of waveforms and the same transmit power. We have also shown that the 2D-DFT can be implemented by a lossless reciprocal MiLAC. Numerical results have confirmed that the MiLAC-aided approach achieves the same CRB and DoA estimation performance as the fully-digital benchmark. Meanwhile, hardware cost and power consumption are reduced because only low-resolution DACs are required at the transmitter, while fewer RF chains and ADCs are needed at the receiver. Moreover, performing the 2D-DFT in the analog domain eliminates all digital DFT operations for DoA estimation.
Future work will examine finite tuning ranges and resolution, component losses, calibration, readout circuitry, bandwidth, and reconfiguration latency. Another direction is an experimental comparison with other analog beamspace architectures such as SIMs and coupler-based DFT networks. We will also explore designing MiLAC to support high-resolution DoA estimation algorithms such as MUSIC.

\section*{Appendix A: Derivation of The Fisher Information Matrices}
\label{appendix_derivative}
In this appendix, we derive the angle-dependent derivatives of the round-trip steering matrix used in the FIM construction. For one target, suppress its angle subscript for readability. Recall that $\mathbf{A}_{k}(\psi_x, \psi_y) = \mathbf{b}(\psi_x, \psi_y) \mathbf{a}^\top(\psi_x, \psi_y) \in \mathbb{C}^{N \times N}$. Since $\mathbf{A}_{k}$ depends on $\psi_x$ and $\psi_y$ in \eqref{eq:psi_xy}, the chain rule gives
\begin{equation}
    \frac{\partial \mathbf{A}_{k}}{\partial \vartheta} = \kappa \cos(\vartheta)\Big[ d_x \cos(\varphi) \frac{\partial \mathbf{A}_{k}}{\partial \psi_x} + d_y \sin(\varphi) \frac{\partial \mathbf{A}_{k}}{\partial \psi_y} \Big],
\end{equation}
\begin{equation}
    \frac{\partial \mathbf{A}_{k}}{\partial \varphi} = \kappa \sin(\vartheta)\Big[ - d_x \sin(\varphi) \frac{\partial \mathbf{A}_{k}}{\partial \psi_x} + d_y \cos(\varphi) \frac{\partial \mathbf{A}_{k}}{\partial \psi_y} \Big].
\end{equation}
Define $\mathbf{D}_{\psi_x} = \jmath (\mathbf{I} \otimes \mathbf{D}_x)$ and $\mathbf{D}_{\psi_y} = \jmath (\mathbf{D}_y \otimes \mathbf{I})$, where $\mathbf{D}_x = \text{diag}(0, 1, \dots, N_x - 1)$ and $\mathbf{D}_y = \text{diag}(0, 1, \dots, N_y - 1)$. The Kronecker-product model yields $\partial\mathbf{a}/\partial\psi_x=\mathbf{D}_{\psi_x}\mathbf{a}$ and $\partial\mathbf{a}/\partial\psi_y=\mathbf{D}_{\psi_y}\mathbf{a}$. The same holds for $\mathbf{b}$. Since $\mathbf{D}_{\psi_x}$ and $\mathbf{D}_{\psi_y}$ are symmetric, the product rule gives
\begin{equation}
    \frac{\partial \mathbf{A}_{k}}{\partial \psi_x} = \mathbf{D}_{\psi_x} \mathbf{A}_{k} + \mathbf{A}_{k} \mathbf{D}_{\psi_x},
\end{equation}
\begin{equation}
    \frac{\partial \mathbf{A}_{k}}{\partial \psi_y} = \mathbf{D}_{\psi_y} \mathbf{A}_{k} + \mathbf{A}_{k} \mathbf{D}_{\psi_y}.
\end{equation}

\bibliographystyle{IEEEtran_url}
\bibliography{IEEEabrv,references}

\end{document}